\documentclass[review]{elsarticle}
\usepackage{lineno,hyperref}
\usepackage{graphicx}
\usepackage{subfigure}
\usepackage{psfrag}
\usepackage{url}
\usepackage{stfloats}
\usepackage{amsmath}
\usepackage{cases}
\usepackage{amssymb}
\usepackage{float}
\usepackage{amsthm}
\usepackage{multirow}
\usepackage{algorithm,algorithmicx}
\usepackage{algpseudocode}
\usepackage[justification=centering]{caption}

\usepackage{color}

\newtheorem{corollary}{Corollary}

\newtheorem{remark}{Remark}
\newtheorem{definition}{Definition}
\newtheorem{lemma}{Lemma}
\newtheorem{theorem}{Theorem}
\newtheorem{example}{Example}
\newtheorem{proposition}{Proposition}

\modulolinenumbers[5]

\journal{Information Sciences}

\begin{document}

\hypersetup{hidelinks}

\begin{frontmatter}

\title{Approximate synchronization of coupled multi-valued logical networks}

\tnotetext[mytitlenote]{This work was supported by the National Natural Science Foundation of China (61877036, 12101366), and the Natural Science Foundation of Shandong Province (ZR2019MF002, ZR2020QF117).}


\author[1]{Rong Zhao}
\ead{zhaorongjy1126@163.com}

\author[1]{Jun-e Feng\corref{mycorrespondingauthor}}
\cortext[mycorrespondingauthor]{Corresponding author}
\ead{fengjune@sdu.edu.cn}

\author[2]{Biao Wang}
\ead{wangbiao@sdu.edu.cn}

\address[1]{School of Mathematics, Shandong University, Jinan 250100, Shandong, P. R. China.}
\address[2]{School of Management, Shandong University, Jinan 250100, Shandong, P. R. China.}

\begin{abstract}
This article deals with the approximate synchronization of two coupled multi-valued logical networks. According to the initial state set from which both systems start, two kinds of approximate synchronization problem, local approximate synchronization and global approximate synchronization, are proposed for the first time. Three new notions: approximate synchronization state set (ASSS), the maximum approximate synchronization basin (MASB) and the shortest approximate synchronization time (SAST) are introduced and analyzed. Based on ASSS, several necessary and sufficient conditions are obtained for approximate synchronization. MASB, the set of all possible initial states, from which the systems are approximately synchronous, is investigated combining with the  maximum invariant subset. And the calculation method of the SAST, associated with transient period, is presented. By virtue of MASB, pinning control scheme is investigated to make two coupled systems achieve global approximate synchronization. Furthermore, the related theories are also applied to the complete synchronization problem of $k$-valued ($k\geq2$) logical networks. Finally, four examples are given to illustrate the obtained results.
\end{abstract}


\end{frontmatter}


\section{Introduction}\label{sec1}

In recent years,  logical dynamic systems have received much attention in the fields of systems biology, physics, medicine and engineering.
Boolean network (BN), introduced by Kauffuman in 1969, is a typical logical dynamic system, which is used to the model genetic regulatory network \cite{KAUFFMAN1969}. In a BN, each gene is regarded as a node of the network and has only two states: 1 or 0, corresponding to on or off, respectively.
Thus, the interaction between genes can be expressed by Boolean functions, and biological processes such as gene regulation and cell differentiation can be well described.
With the development of Cheng product, also called semi-tensor product (STP) \cite{Chengdaizhan2011c}, a great number of research results about BNs have been obtained, such as  controllability and observability\cite{
2018Controllability, lujianquan2021observability}, stability and stabilization\cite{huangchi2020,2020Feedbackstabilization}, optimal control\cite{WU2018optimal,WU2019optimal}, disturbance decoupling problem\cite{2017disturbancedecoupling,wang2021DisturbanceDecoupling}, detectability \cite{wangbiao2019detectabilityPBN,wangbiao2020detectabilityBN,chenzengqiang2021} and so on.

However, each state variable of a BN only takes two possible values. Therefore, BNs can not describe many actual situation accurately. For example, in a game process, each player has multiple strategy choices, which cannot be accurately described by a BN. Therefore, a more general multi-valued logical network (MVLN), also called $k$-valued logical network, comes into being. By resorting to Cheng product, basic properties of MVLNs have been well studied in \cite{lizhiqiang2010}\cite{cheng2012}.
Although an MVLN can be regarded as a natural generalization of a BN, it still has more complicated structure and wider applications than BNs. Research works related to MVLN have been widely applied to genetic regulatory networks \cite{Kitano2009Foundations}, finite automata\cite{yuejumei2019automata, zhangzhipeng2020automata}, networked evolutionary games\cite{fushihua2017games,Maoying2018games}, nonlinear feedback shift registers\cite{ liuzhenbin2015register, wanghaiyan2017registers,lu2021registers} and other fields. Hence, it is necessary and significant to further study MLVNs.

Synchronization is a quite common collective behavior phenomenon in real world, such as the synchronous glow of fireflies, birds flying in the same direction, synchronous transmission of signals, laser oscillation synchronization, harmonious fluctuations of heart muscle cells and brain neural networks \cite{Ditto2002synchronization}.
Thus it can be seen that synchronization phenomenon has comprehensive research and application background in the fields of physics, biology, chemistry, engineering and so on.
In the past few decades, the research of clustering and synchronization in networks has drawn much attention of scholars, such as synchronization of complex networks \cite{Arenas2008complexnetwoks},
$H_\infty$ synchronization of uncertain neural networks \cite{Karimi2010HinfinitySynchronization}, synchronization of Kauffman networks \cite{Morelli2001Kauffmannetworks} and so on.
Recently, the synchronization of logical networks has become a hot topic, and different kinds of synchronization problems have been well investigated.
Complete synchronization of drive-response BNs was studied in \cite{lirui2012CompleteSynchronization} firstly.
Partial synchronization and local synchronization of interconnected BNs were investigated in \cite{chenhongwei2017partialsynchronization} and \cite{chenhongwei2020LocalSynchronization}, respectively.
Cluster synchronization of BNs was discussed in \cite{zhanghao2018cluster} via open-loop control. Robust synchronization of BNs with disturbances was studied in \cite{liyuanyuan2018robustsynchronization}. Impulsive effects on synchronization was considered in \cite{zhongjie2014}. Synchronization of  an array of output-coupled BNs was investigated in \cite{lujianquan2021}.
In addition, the synchronization of MVLNs was also considered in \cite{mengmin2014multi-valued, liyalu2018event-triggeredcontrol}.

On the other hand, many classical control schemes have been extensively investigated to settle the synchronization problem of BNs, such as state feedback control\cite{liuyang2016feedbackcontrol}, event-triggered control \cite{liyalu2018event-triggeredcontrol} and sampled-data control \cite{liuyang2019Sampled-datacontrol}.
In particular, by taking advantage of pinning control, reference \cite{lifangfei2016pinningcontrol} solved the synchronization problem of drive-response BNs.
It is worth noting that pinning control is an effective control strategy in control theory.
The basic idea of pinning control is to adjust the properties of a part of nodes in the network, such that the whole network can achieve expected behavior through the mutual coupling between nodes.
Compared with other control methods, the advantage of pinning control is that it can dramatically reduce the energy consumption since it only needs to control a fraction of  nodes.
Moreover, reference \cite{Rosin2013} showed that synchronization patterns can be achieved by adjusting the refractory time of 2 out of 32 nodes, which means that one can achieve control goals of the global network via regulating quite a small fraction of nodes.

It is worth pointing out that complete synchronization is impracticable in many circumstances, such as the existence of parameter mismatches, the nonidentical node dynamics in heterogeneous networks. But the synchronization errors are expected to  be restricted  in a small permissible range, that is approximate synchronization\cite{xujunqun2013approxiamtecomplex}.
In complex networks, chaotic systems and other fields, the approximate synchronization problem has drawn much attention, such as \cite{xujunqun2013approxiamtecomplex,Sorrentino2016approximate,2021Quasi,2020Bounded}.
Due to various possible deviations, natural imperfection, human interference and other uncontrollable factors in practical engineering, the synchronous error is also inevitable in logical networks, which attracts us to propose a general framework to deal with these situations. In addition, to the best of our knowledge, up to now, there is no literature on approximate synchronization of logical systems.  Motivated by these, we pay attention to the approximate synchronization of coupled MVLNs.

In this paper, the pinning control strategy is employed to settle the approximate synchronization of coupled MVLNs. The main challenges of this paper include:
\begin{enumerate}[1)]
  \item How to propose a reasonable concept for approximate synchronization by virtue of the distinguish of coupled MVLNs is the first challenge. Different from complex networks, logical networks have special structure and properties, then relevant approximate synchronization problem in \cite{xujunqun2013approxiamtecomplex,Sorrentino2016approximate,2021Quasi,2020Bounded} cannot be directly applied to logical networks.

  \item How to select pinning nodes is one of the difficulties. Different from [39], where the systems considered are drive-response BNs and the pinning controllers are injected in response network, our research objectives are coupled MVLNs, which is more complicated than drive-response BNs.
  \item How to construct the pinning feedback controllers after getting pinning nodes is also challenging to address. In [39], pinning feedback controllers are derived by solving logic matrix equations based on enumeration method, which is rather tedious. Therefore, giving criteria and corresponding algorithm for constructing the pinning feedback controllers is also a challenge.
\end{enumerate}

The main contributions of this paper are presented as follows.
\begin{enumerate}[1)]
  \item Two novel approximate synchronization concepts of coupled MVLNs, including local version and global version, are put forward for the first time. Compared with \cite{lirui2012CompleteSynchronization, mengmin2014multi-valued, chenhongwei2020LocalSynchronization}, the approximate synchronization discussed in this paper is more general and can degenerate into complete synchronization.
  \item The approximate synchronous state set and the maximum approximate synchronization basin (MASB) are investigated. Based on these, several necessary and sufficient conditions are provided for determining global and local approximate synchronization.
      Besides, the shortest approximate synchronization time is introduced and its calculation method is provided resorting to the maximum invariant subset.
  \item Pinning control strategy is considered to make coupled MVLNs achieve global approximate synchronization. Based on the MASB and the maximum invariant subset, pinning nodes are found out. Further, criteria for the solvability of pinning feedback controllers are presented. Compared with \cite{lifangfei2016pinningcontrol}, our approach improves and extends the results in \cite{lifangfei2016pinningcontrol}, and thus is more general and applicable.
\end{enumerate}

The remainder of this paper is organised as follows. At the end of Section \ref{sec1}, some notations, used in this paper, are presented. Section \ref{sec2} introduces Cheng product and $k$-valued logic.  Section \ref{sec3} proposes and briefly analyzes the approximate synchronization problem.
Section \ref{sec4} is the main results of this paper, including the approximate synchronization conditions, the calculation of MASB and the shortest approximate synchronization time, and the design of pinning control strategy.
In Section \ref{sec6}, related results are applied to the complete synchronization problem of $k$-valued logical networks.
Section \ref{sec5} gives some examples to verify the validity of the obtained results.
Section \ref{sec7} is a brief conclusion of this paper.

$Notations:$
\begin{enumerate}[1)]
  \item $\mathbb{Z_+}$: the set of all positive integers.
  \item $\mathcal{D}_k:=\{0, \frac{1}{k-1},\ldots,\frac{k-2}{k-1}, 1\}$, $k\geq2$. In particular, $\mathcal{D}:=\{0, 1\}$.
  \item $\Delta_k:=\{\delta_k^i\mid i=1,2,\ldots,k\}$, where $\delta_k^i$ represents the $i$-th column of identity matrix $I_k$. Denote $\Delta_2:=\Delta$.
  \item $\mathbb{R}^{m\times n}$: the set of $m\times n$ real matrices.
  \item $Col_i(M)$: the $i$-th column of matrix $M$.
  \item $Col(M)$ : the set of columns of $M$.
  \item $[M]_{i,j}$: the element on the $(i,j)$ entry of matrix $M$.
  \item $sgn(a):=
  \begin{cases}
  1,~&a>0,\\
  0,~&a=0,\\
  -1,~&a<0.
 \end{cases}$ If $M=([M]_{i,j})_{m\times n}\in\mathbb{R}^{m\times n}$, then $sgn(M):=(sgn[M]_{i,j})_{m\times n}$.
  \item $\mathcal{L}_{m\times n}$: the set of $m\times n$ logic matrices. $L\in\mathcal{L}_{m\times n}$  means $Col(L)\subseteq\Delta_m$.
  \item If $L\in\mathcal{L}_{m\times n}$, then it can be expressed as $L=[\delta_m^{i_1}, \delta_m^{i_2}, \ldots , \delta_m^{i_n}]$. For the sake of compactness, it is briefly denoted by $L=\delta_m[i_1, i_2, \ldots, i_n]$.
  \item $\mathcal{B}_{m\times n}$: the set of $m\times n$ Boolean matrices. $M\in\mathcal{B}_{m\times n}$ means that all its entries are either 0 or 1.
  \item $\otimes$: Kronecker product of matrices.
  \item $\ast$: Khatri-Rao product of matrices.
  \item $\circ$: Hadamard product of matrices.
  \item $M^\top$: the transpose of  matrix $M$.
  \item $\mathbf{1}_n:=[\underbrace{1, 1, \ldots, 1}_n]^\top$.
  \item $|S|$: the cardinality of set $S$.
  \item For two vectors $V_1=[v_{11},v_{12},\ldots,v_{1n}]^\top$ and $V_2=[v_{21},v_{22},\ldots,v_{2n}]^\top$, the inequality $V_1\leq V_2$ means that $v_{1j}\leq v_{2j}$, $j=1,2,\ldots,n$.
\end{enumerate}

\section{Preliminaries}\label{sec2}
In this section, we introduce some necessary preliminaries about Cheng product and $k$-valued logic.

\subsection{Cheng Product}
To begin with, the definition of Cheng product is given, and some useful properties are proposed.
For details, please refer to \cite{Chengdaizhan2011c}.
\begin{definition}{\rm\cite{Chengdaizhan2011c}}
 Let $A\in\mathbb{R}^{m\times n}$ and  $B\in\mathbb{R}^{p\times q}$. The Cheng product of $A$ and $B$  is defined as
\begin{equation}
    A\ltimes B=(A\otimes I_{t/n})(B\otimes I_{t/p}),
\end{equation}
where $t$ is the least common multiple of $n$ and $p$.
\end{definition}

Note that when $n=p$, Cheng product degenerates into the traditional matrix product. Without confusion, the symbol $\ltimes$ is omitted throughout this paper.

\begin{lemma}{\rm\cite{Chengdaizhan2011c}}\label{lem1}
Three basic properties of Cheng product are given as follows.
\begin{enumerate}[1)]
  \item Let $A\in\mathbb{R}^{m\times n}$ and $X\in\mathbb{R}^{t\times1}$. Then
    \begin{equation}
      XA=(I_t\otimes A)X.
    \end{equation}
  \item  Let $X\in\mathbb{R}^{m\times 1}$ and $Y\in\mathbb{R}^{n\times 1}$ be two column vectors. Then
    \begin{equation}
      W_{[m,n]}XY=YX,
    \end{equation}
    where $W_{[m,n]}:=$ $[I_n\otimes\delta_m^1, I_n\otimes\delta_m^2, \ldots, I_n\otimes\delta_m^m]$ is called the swap matrix.
  \item Let $x\in\Delta_k$. Then
    \begin{equation}
     x^2=M_{r,k}x,
    \end{equation}
    where $M_{r,k}:=$ $diag\{\delta_k^1, \delta_k^2,\ldots,\delta_k^k\}$ is called the power-reducing matrix.
\end{enumerate}
\end{lemma}
Notice that 1) of Lemma 1 reveals that Cheng product satisfies pseudo commutativity. 2) of Lemma 1 shows that two vectors can be exchanged through the swap matrix. And 3) of Lemma 1 illustrates the order reduction effect of power-reducing matrix.

\subsection{$k$-Valued Logic}

This subsection introduces $k$-valued logic and gives the matrix expression of $k$-valued logic.
\begin{definition}{\rm\cite{Chengdaizhan2011c}}
A variable $x$ is called a $k$-valued logical variable if it takes a value from $\mathcal{D}_k$, i.e., $x\in\mathcal{D}_k$.
\end{definition}
Some useful logical operators of $k$-valued logic are presented as follows.
\begin{definition}{\rm\cite{Chengdaizhan2011c}}
Assume $a$ and $b$ are two $k$-valued logical variables. Then
\begin{enumerate}[1)]
  \item Negation $(\neg):$ $\neg a=1-a$.
  \item Conjunction $(\wedge):$ $a\wedge b=\min\{a,b\}$.
  \item Disjunction $(\vee):$ $a\vee b=\max\{a,b\}$.
  \item Mod $k$ addition $(\oplus_k):$ $a\oplus_k b=\frac{[(k-1)(a+b)](mod~k)}{k-1}$.
  \item The i-confirmor $(\nabla_{i,k}):$ $\nabla_{i,k}(a)=
  \begin{cases}
  1,~a=\frac{k-i}{k-1},\\
  0,~otherwise.
  \end{cases}$
  \item The rotator $(\oslash_k):$ $\oslash_k(a)=
  \begin{cases}
  a-\frac{1}{k-1}, &a\neq0,\\
  1, &a=0.
  \end{cases}$
\end{enumerate}
\end{definition}

In order to obtain the matrix expression of the $k$-valued logical operators, we identify $\frac{k-i}{k-1}$ as $\delta_{k}^i$, $i=1,2,\ldots,k$. With Cheng product, we have the following lemma.

\begin{lemma}{\rm\cite{Chengdaizhan2011c}}\label{lem2}
\begin{enumerate}[1)]
\item Given a $k$-valued logical function $f: \mathcal{D}_k^n\rightarrow\mathcal{D}_k$. There exists a unique matrix $M_f\in\mathcal{L}_{k\times k^n}$ such that the vector form of $f$ is
\begin{equation}\label{eqlem2}
f(x_1, x_2, \ldots, x_n)=M_f\ltimes_{i=1}^n x_i,
\end{equation}
where $x_i\in\Delta_k$ and $M_f$ is called the structure matrix of $f$.
\item  Assume a $k$-valued logical function $f$ has algebraic form $\eqref{eqlem2}$.
Then the logical form of $f$ can be expressed as
\begin{align}
&f=[\nabla_{1,k}(x_1)\wedge f_1(x_2,\ldots,x_n)]\vee[\nabla_{2,k}(x_1)\wedge f_2(x_2,\notag\\
&\ldots,x_n)] \vee\cdots\vee[\nabla_{k,k}(x_1)\wedge f_k(x_2,\ldots,x_n)],
\end{align}
where $f_i$ has $Blk_i(M_f)$ as its structure matrix, $i=1,2,\ldots,k$, and  $Blk_i(M_f)\in\mathcal{L}_{k\times k^{n-1}}$ is the $i$-th block by splitting $M_f$ equally into $k$ blocks.
\end{enumerate}
\end{lemma}


According to Lemma \ref{lem2}, some useful structure matrices of common logical operators are provided as follows. Denote the structure matrices of $\neg$, $\wedge$, $\vee$, $\oplus_k$, $\nabla_{i,k}$ and $\oslash_k$ by $M_{n,k}$, $M_{c,k}$, $M_{d,k}$, $M_{\oplus_k}$, $M_{\nabla_{i,k}}$ and $M_{\oslash_k}$, respectively. Then

\begin{enumerate}[1)]
  \item[1)] $M_{n,k}=\delta_k[k,k-1,\ldots,1]\in\mathcal{L}_{k\times k}$.
  \item[2)] $M_{c,k}=\delta_{k}[\underbrace{1,2,3,\ldots,k}_k,
\underbrace{2,2,3,\ldots,k}_k,\ldots,\underbrace{k,k,\ldots,k}_k]\in\mathcal{L}_{k\times k^2}$.
  \item[3)] $M_{d,k}=\delta_{k}[\underbrace{1,1,1\ldots,1}_k,
\underbrace{1,2,2,\ldots,2}_k,\ldots,\underbrace{1,2,\ldots,k}_k]\in\mathcal{L}_{k\times k^2}$.
  \item[4)]
$M_{\oplus_k}=\delta_k[\underbrace{2,3,\ldots,k,1}_k,
\underbrace{3,\ldots,k,1,2}_k,\ldots,\underbrace{1,2,\ldots,k}_k]\in\mathcal{L}_{k\times k^2}$.
  \item[5)]
$M_{\nabla_{i,k}}=\delta_k[\underbrace{k,k,\ldots,k}_{i-1},1,\underbrace{k,k,\ldots,k}_{k-i}]\in\mathcal{L}_{k\times k}$.
  \item[6)]
$M_{\oslash_k}=\delta_{k}[2,3,4,\ldots,k,1]\in\mathcal{L}_{k\times k}$.
\end{enumerate}

\section{Problem Formulation}\label{sec3}

In this section, we put forward the approximate synchronization concept of $k$-valued ($k>2$) logical networks.

Consider the following two coupled $k$-valued ($k>2$) logical networks.
\begin{align}
  x_i(t+1)=f_i(x_1(t),\ldots,x_n(t),z_1(t),\ldots,z_n(t)),\label{sys1}\\
  z_i(t+1)=g_i(x_1(t),\ldots,x_n(t),z_1(t),\ldots,z_n(t)),\label{sys2}
\end{align}
where $x_i(t)\in\mathcal{D}_k$, $z_i(t)\in\mathcal{D}_k$, $i=1,2,\ldots,n$ are state variables of the coupled systems \eqref{sys1} and \eqref{sys2}, respectively,
and $f_i: \mathcal{D}_k^n\rightarrow\mathcal{D}_k$,
$g_i: \mathcal{D}_k^n\rightarrow\mathcal{D}_k$, $i=1,2,\ldots,n$ are logical functions.
Let $X(t):=(x_1(t),\ldots,x_n(t))\in\mathcal{D}_k^n$ and $Z(t):=(z_1(t),\ldots,z_n(t))\in\mathcal{D}_k^n$.
The state trajectories of systems \eqref{sys1} and \eqref{sys2} starting from initial states $(X_0, Z_0)\in\mathcal{D}_k^{2n}$ are denoted by $X(t; X_0,Z_0)$ and $Z(t;X_0,Z_0)$ with  $x_i(t;X_0,Z_0)$ and $z_i(t;X_0,Z_0)$ as their $i$-th components, respectively.

Next, we propose the definition of approximate synchronization of $k$-valued ($k>2$) logical networks.
\begin{definition}\label{def1}
{\rm (Approximate Synchronization Definitions)}
\begin{enumerate}[1)]
\item Given a state set $\Psi\subseteq\mathcal{D}_k^{2n}$. Systems \eqref{sys1} and \eqref{sys2} are called locally approximately synchronous with respect to $\Psi$, if there exists  an  integer $\rho\in\mathbb{Z_+}$ such that
\begin{equation}\label{eqdef1}
\max_{1\leq i\leq n}|x_i(t;X_0,Z_0)-z_i(t;X_0,Z_0)|\leq \frac{1}{k-1}
\end{equation}
holds for any $(X_0, Z_0)\in\Psi$ and $t\geq \rho$.
\item Systems \eqref{sys1} and \eqref{sys2} are called globally approximately synchronous, if there exists an integer $\rho\in\mathbb{Z_+}$ such that
\begin{equation}\label{eqdef2}
\max_{1\leq i\leq n}|x_i(t;X_0,Z_0)-z_i(t;X_0,Z_0)|\leq \frac{1}{k-1}
\end{equation}
holds for any $(X_0, Z_0)\in\mathcal{D}_k^{2n}$ and $t\geq \rho$.
\end{enumerate}
\end{definition}

$\Psi\subseteq\mathcal{D}_k^{2n}$ is called an $approximate$ $synchronization$ $basin$ of systems \eqref{sys1} and \eqref{sys2}, if  \eqref{sys1} and \eqref{sys2} are approximately synchronous with respect to $\Psi$. The $maximum$ $approximate$ $synchronization$ $basin$ (MASB) of systems \eqref{sys1} and \eqref{sys2} is denoted by $\Psi_{\max}$.
Denote the vector form of $\Psi\subseteq\mathcal{D}_k^{2n}$ by $\Phi\subseteq\Delta_{k^{2n}}$, i.e., $\Psi\sim\Phi$. Similarly, $\Psi_{\max}\sim\Phi_{\max}$.
From 1) of Definition \ref{def1}, if systems \eqref{sys1} and \eqref{sys2} are approximately synchronous with respect to $\Psi\subseteq\mathcal{D}_{k}^{2n}$, then there exists an  integer $\rho\in\mathbb{Z_+}$ such that \eqref{eqdef1} holds.
On this ground, the minimum integer $\rho\in\mathbb{Z_+}$ satisfying 1) of Definition \ref{def1}, denoted by $\Gamma_{\Psi}$, is called the $shortest$ $approximate$ $synchronization$ $time$ with respect to $\Psi$. Similarly, the minimum integer $\rho\in\mathbb{Z_+}$ satisfying 2) of Definition \ref{def1}, denoted by $\Gamma$, is called the $global$ $shortest$ $approximate$ $synchronization$ $time$.

\begin{remark}\label{rem1}

1) (Generalization of Definition \ref{def1})

In Definition \ref{def1}, conditions \eqref{eqdef1} and \eqref{eqdef2} can be generalized as
\begin{equation}\label{eqrem1}
\max_{1\leq i\leq n}|x_i(t;X_0,Z_0)-z_i(t;X_0,Z_0)|\leq \frac{\gamma}{k-1},
\end{equation}
where $0\leq\gamma\leq k-1$. Depending on practical requirements, $\gamma$ can take different values. The value of $\gamma$  reflects actually the degree of approximate synchronization. It is obvious that the smaller the value of $\gamma$, the higher the degree of approximate synchronization. Obviously, if $\gamma=0$, then complete synchronization occurs.
Even if the complete synchronization is unable to realize, we can still predict and investigate the dynamic behaviors  and characteristics of the coupled systems since the upper bound of the synchronization error can be estimated. Thus it will has potential and wide-ranging application prospects, such
as in the fields of multi-agent leader-follower systems\cite{multi-agent}, biochemical systems\cite{Heidel2003biochemicalsystems}, hierarchical networked
evolutionary games\cite{2004Evolutionary}, etc.

In this paper, we mainly investigate the approximate synchronization for the case of $\gamma=1$.
It is worth pointing out that if $k=2$, then each node has only two possible values: 0 and 1. In this case, it is meaningless to consider approximate synchronization. Hence, we only focus on the case of $k>2$ in regard to approximate synchronization.

2) (Physical meaning of the shortest approximate synchronization time)

The shortest approximate synchronization time defined as the smallest integer satisfying Definition 4, can reflect the speed of approximate synchronization of two logical networks, which also represents the strength of the synchronization capability of two systems from the perspective of optimal time. Hence it is meaningful in the process of theoretical analysis and engineering practice.
\end{remark}

According to Lemma \ref{lem2}, we assume that $F_i\in\mathcal{L}_{k\times k^{2n}}$ and $G_i\in\mathcal{L}_{k\times k^{2n}}$ are the structure matrices of $f_i$ and $g_i$, $i=1,2,\ldots,n$, respectively,
then systems \eqref{sys1} and \eqref{sys2} can be converted into the following equivalent algebraic form:
\begin{align}
  x(t+1)=Fx(t)z(t),\label{sys3}\\
  z(t+1)=Gx(t)z(t),\label{sys4}
\end{align}
where $x(t):=\ltimes_{i=1}^nx_i(t)\in\Delta_{k^n}$, $z(t):=\ltimes_{i=1}^nz_i(t)\in\Delta_{k^n}$, $F:=F_1\ast F_2*\cdots\ast F_n\in\mathcal{L}_{k^n\times k^{2n}}$,
$G:=G_1\ast G_2*\cdots\ast G_n\in\mathcal{L}_{k^n\times k^{2n}}$.
Define $\xi(t)=x(t)z(t)\in\Delta_{k^{2n}}$, then we can get the following augmented system:
\begin{equation}\label{sys5}
  \xi(t+1)=L\xi(t),
\end{equation}
where $L:=F(I_{k^{2n}}\otimes G)M_{r,k^{2n}}\in\mathcal{L}_{k^{2n}\times k^{2n}}$ is the transition matrix of system \eqref{sys5}. The trajectory of system \eqref{sys5} starting  from initial state $\xi_0\in\Phi$ is denoted by $\xi(t;\xi_0)$.

From the analysis above, \eqref{sys1} and \eqref{sys2} are converted into an augmented system \eqref{sys5}, which is helpful to systematically analyze the synchronization problem.
Next, we analyze which states satisfy the approximate synchronization condition \eqref{eqdef1} or \eqref{eqdef2}.

Suppose the states of systems \eqref{sys1} and \eqref{sys2} are expressed as:
\begin{align}
X(t)=\left(\frac{k-i_1}{k-1}, \frac{k-i_2}{k-1},\ldots,\frac{k-i_n}{k-1}\right),\label{eq-1}\\
Z(t)=\left(\frac{k-j_1}{k-1}, \frac{k-j_2}{k-1},\ldots,\frac{k-j_n}{k-1}\right),\label{eq-2}
\end{align}
where $i_l, j_l\in\{1,2,\ldots,k\}$, $l=1,2,\ldots,n$.
From Definition \ref{def1}, \eqref{eqdef1} and \eqref{eqdef2} implies that $|i_l-j_l|\leq 1$, $l=1,2,\ldots,n$ when $t\geq\rho$.
Let
\begin{equation}\label{eq1}
\begin{cases}
 i=i_n+\sum_{l=1}^{n-1}(i_{n-l}-1)k^l,\\
 j=j_n+\sum_{l=1}^{n-1}(j_{n-l}-1)k^l,
\end{cases}
\end{equation}
where $|i_l-j_l|\leq 1$, $i_l, j_l\in\{1,2,\ldots,k\}$, $l=1,2,\ldots,n$.
Construct
\begin{equation}\label{eq2}
\Lambda=\{\delta_{k^{2n}}^{(i-1)k^n+j}: ~i,j~\text{satisfy}~\eqref{eq1}\}.
\end{equation}
$\Lambda$ is called the $approximate~synchronous~state~set$ of system \eqref{sys5}.
Then the following result is immediate.
\begin{corollary}\label{th0}
Consider systems \eqref{sys1} and \eqref{sys2} with the augmented system \eqref{sys5}.
\eqref{sys1} and \eqref{sys2} are approximately synchronous with respect to $\Psi$, if and only if there exists an  integer $\rho\in\mathbb{Z_+}$ such that $\xi(t;\xi_0)\in\Lambda$ holds for any $t\geq\rho$ and $\xi_0\in\Phi$.
\end{corollary}

The following  proposition illustrates the number of elements in $\Lambda$.
\begin{proposition}\label{pro1}
The cardinality of $\Lambda$ is
\begin{equation}\label{eqpro1}
|\Lambda|=\sum_{l=0}^n\mathbf{C}_n^l3^l(k-2)^l2^{2(n-l)},
\end{equation}
where $\mathbf{C}_n^l$ represents the number of all combinations of taking $l$ elements from $n$ different elements.
\end{proposition}
\begin{proof}
It is noted that with the representation \eqref{eq-1} and \eqref{eq-2}, $i_l$ and $j_l$ need to meet the condition of $|i_l-j_l|\leq1$, $i_l,j_l\in\{1,2,\ldots,k\}$, $l=1,2,\ldots,n$. Therefore, for each $i_l$, $l=1,2,\ldots,n$, if $i_l=1$ or $i_l=k$, then $j_l$ has two possible cases, i.e., $j_l\in\{1,2\}$ or $j_l\in\{k-1,k\}$. If $1<i_l<k$, then $j_l$ has three possible cases, which means $j_l\in\{i_l-1,i_l,i_l+1\}$.
According to \eqref{eq1}, $i$ is determined by $i_1, i_2, \ldots, i_n$, and $j$ is determined by $j_1, j_2, \ldots, j_n$. Suppose that there are $\theta$ components in $i_1, i_2, \ldots, i_n$  satisfying $1<i_l<k$, then the value of remaining $n-\theta$ components are $1$ or $k$. At this time, the number of elements included in $\Lambda$ is
\begin{align*}
&\mathbf{C}_n^\theta(k-2)^\theta 2^{(n-\theta)}3^\theta2^{(n-\theta)}\\
=&\mathbf{C}_n^\theta3^\theta(k-2)^\theta 2^{2(n-\theta)}.
\end{align*}
When $\theta$ goes through $0,1,2,\ldots,n$, all elements of $\Lambda$ are obtained. Hence, \eqref{eqpro1} is clear.
\end{proof}

As Theorem \ref{th0} shows that the approximate synchronous state set $\Lambda$ plays an important role in studying approximate synchronization problem. The subsequent theoretical results of this paper are mainly proposed based on $\Lambda$.

\section{Main Results}\label{sec4}

In this section, we study the approximate synchronization of systems \eqref{sys1} and \eqref{sys2}, and present the main results of this paper. First, some criteria for approximate synchronization are proposed. Second, the methods for calculating the MASB and the shortest approximate synchronization time are presented.  Third, pinning control strategy is considered to make sure that the two coupled MVLNs achieve global approximate synchronization.

\subsection{Approximate Synchronization Conditions}\label{sec4-1}

In this subsection, several criteria for approximate synchronization are presented.

Notice that the state space $\Delta_{k^{2n}}$ of system \eqref{sys5} is finite. The trajectory of system \eqref{sys5} starting from any initial state will eventually stay in some states or cycles\cite{Chengdaizhan2011c, cheng2010}. So, we introduce the following definition firstly.
\begin{definition}{\rm\cite{cheng2010}}
Consider system \eqref{sys5}.
\begin{enumerate}[1)]
\item A state $\xi\in\Delta_{k^{2n}}$ is called a fixed point of system \eqref{sys5}, if $L\xi=\xi$.
\item $\{\xi_0, \xi_1,\ldots,\xi_l\}$ is called a limit cycle of system \eqref{sys5} with length $l+1$, if $\xi_{i+1}=L\xi_i$, $i=0,1,\ldots,l-1$ and $L\xi_l=\xi_0$.
\end{enumerate}
\end{definition}

Both fixed point and limit cycle are called $attractor$. Denote by $\Omega$ the union of all attractors of system \eqref{sys5}.
For a state $\xi_0$, its $transient$ $period$, denoted by $\tau_{\xi_0}$, is the smallest integer $t$ such that $\xi(t;\xi_0)\in\Omega$, where $\Omega_{\xi_0}:=\{\xi(t;\xi_0): t\geq\tau_{\xi_0}\}$ is called the attractor of $\xi_0$.
For a state set $\Phi\subseteq\Delta_{k^{2n}}$,  $\Omega_\Phi:=\mathop{\bigcup_{\xi_0\in\Phi}}\Omega_{\xi_0}$ represents the set of all attractors of $\Phi$.
The transient period of system \eqref{sys5}, denoted by $\tau$, is defined as:
\begin{equation}\label{eq3}
  \tau=\max_{\xi\in\Delta_{k^{2n}}}\tau_{\xi}.
\end{equation}
Assume
$\Omega=\{\mathcal{C}_1,\mathcal{C}_2,\ldots,\mathcal{C}_\epsilon\}$,
where the length of $\mathcal{C}_i$ is denoted by $l_i$, $i\in\{1,2,\ldots,\epsilon\}$.
(Using STP toolbox\footnote{http://lsc.amss.ac.cn/$\sim$dcheng/stp/STP} in Matlab, the attractors of a given logical network is easy to obtained.)

The following lemma, introduced in \cite{cheng2010} firstly, is useful in analyzing the approximate synchronization problem.
\begin{lemma}\label{lem3}{\rm\cite{cheng2010}}
Consider system \eqref{sys5}. The transition matrix $L$ satisfies
$L^\tau=L^{\tau+\lambda}$, where $\tau$ is the transient period of system \eqref{sys5} and $\lambda=lcm\{l_1,l_2,\ldots,l_\epsilon\}$ is the least common multiple of $l_1,l_2,\ldots,l_\epsilon$.
\end{lemma}

Based on Lemma \ref{lem3}, the following theorem presents two necessary and sufficient conditions for the approximate synchronization problem.
\begin{theorem}\label{th1}
Consider systems \eqref{sys3} and \eqref{sys4}. Suppose $\tau$ is the transient period of system \eqref{sys5}.
\begin{enumerate}[1)]
\item Systems \eqref{sys3} and \eqref{sys4} are approximately synchronous with respect to $\Phi\subseteq\Delta_{k^{2n}}$, if and only if
\begin{equation}\label{eq4}
  \bigcup_{t=\tau}^{\tau+\lambda-1}\{Col_i(L^t): \delta_{k^{2n}}^i\in\Phi\}\subseteq\Lambda,
\end{equation}
where  $\lambda$ is defined in Lemma \ref{lem3}.
\item  Systems \eqref{sys3} and \eqref{sys4} are globally approximately synchronous, if and only if
\begin{equation}
 Col(L^\tau)\subseteq\Lambda.
\end{equation}
\end{enumerate}
\end{theorem}
\begin{proof}
See Appendix.
\end{proof}

In order to give a more concise criterion, we define the index vectors of sets $\Phi$ and $\Lambda$ in the form of
\begin{align}
\Xi_1=\sum_{\delta_{k^{2n}}^i\in\Phi}\delta_{k^{2n}}^i\in\mathcal{B}_{k^{2n}\times1},\label{eq7}\\
\Xi_2=\sum_{\delta_{k^{2n}}^i\in\Lambda}\delta_{k^{2n}}^i\in\mathcal{B}_{k^{2n}\times1}.\label{eq7-1}
\end{align}

Subsequently, it is not hard to get the following conclusions in light of Theorem \ref{th1}, which is effective to determine whether two systems are approximately synchronous.
\begin{corollary}\label{th2}
Consider systems \eqref{sys3} and \eqref{sys4}. Suppose $\tau$ is the  transient period of system \eqref{sys5}.
\begin{enumerate}[1)]
\item  Systems \eqref{sys3} and \eqref{sys4} are approximately synchronous with respect to $\Phi$, if and only if
\begin{equation}\label{eqth2}
sgn\left(\left(\sum_{t=\tau}^{\tau+\lambda-1}L^t\right)\Xi_1\right)\leq\Xi_2,
\end{equation}
where  $\lambda$ is defined in Lemma \ref{lem3}.
\item Systems \eqref{sys1} and \eqref{sys2} are globally approximate synchronous, if and only if
\begin{equation}
sgn\left(L^\tau\mathbf{1}_{k^{2n}}\right)\leq\Xi_2.
\end{equation}
\end{enumerate}
\end{corollary}
\begin{proof}
See Appendix.
\end{proof}

\subsection{The Maximum Approximate Synchronization Basin and the Shortest Approximate Synchronization Time}\label{sec4-2}

In this subsection, we provide an approach to find the MASB $\Phi_{\max}$. Besides, the issue of the shortest approximate synchronization time is solved.

Combined with Corollary \ref{th2}, an algorithm (Algorithm 1) is derived to calculate the MASB $\Phi_{\max}$. (Please see Appendix for the proof of the validity of Algorithm 1.)
\begin{algorithm}[H]\label{alg0}
\caption{Calculate the MASB $\Phi_{\max}$.}
\begin{algorithmic}[1]
\Require $L$, $\Lambda$
\Ensure $\Phi_{max}$
\State Initialize $\Phi_{\max}=\Delta_{k^{2n}}$;
\State Compute the transient period $\tau$  by Lemma \ref{lem3};
\State Compute the index vector $\Xi_2$ of $\Lambda$ by \eqref{eq7-1};
\For{$i=1 \to k^{2n}$}
\State Compute $\sum_{t=\tau}^{\tau+\lambda-1}Col_i(L^t)$;
\If {$sgn\left(\sum_{t=\tau}^{\tau+\lambda-1}Col_i\left(L^t\right)\right)\nleq\Xi_2$}
\State $\Phi_{\max}=\Phi_{\max}\setminus\{\delta_{k^{2n}}^i\}$;
\EndIf
\EndFor
\end{algorithmic}
\end{algorithm}

\begin{remark}
In Algorithm 1, the computational complexity of lines 2 and 3 are $O(k^{6n})$ and $O(k^{2n})$, respectively. And the computational complexity of the ``for" loop in lines 4-9 is $O(k^{8n})$. Therefore, the computational complexity of Algorithm 1 is $O(k^{8n})$.
\end{remark}
In light of the definition of $\Phi_{\max}$, the following conclusions are clear.
\begin{corollary}\label{cor1}
Consider  systems \eqref{sys3} and \eqref{sys4}.
\begin{enumerate}[1)]
  \item Systems \eqref{sys3} and \eqref{sys4} are not approximately synchronous with respect to any initial state set if and only if $\Phi_{\max}=\emptyset$.
  \item Systems \eqref{sys3} and \eqref{sys4} are globally approximately synchronous if and only if $\Phi_{\max}=\Delta_{k^{2n}}$.
\end{enumerate}
\end{corollary}

In the following, we introduce the concept of maximum invariant subset, which is useful to further describe the relationship with $\Phi_{\max}$, calculate the shortest approximate synchronization time, and design the pinning controllers.

A set $S\subseteq\Delta_{k^{2n}}$ is called an $invariant$ $subset$ of system \eqref{sys5}, if $\xi(t;\xi_0)\in S$ holds for any $t\geq0$ and $\xi_0\in S$.
It is clear  that  $S$ is an invariant subset of system \eqref{sys5} if and only if $L_S\subseteq S$, where $L_S:=\{Col_i(L): \delta_{k^{{2n}}}^i\in S\}$.
Note that the union of two invariant subsets is still an invariant subset. The union of all invariant subsets contained in a given set $M$ is called the $maximum$ $invariant$ $subset$ of $M$, denoted by $\mathcal{I}_m^\Lambda$.

Based on the analysis above, an algorithm (Algorithm 2) is  put forward to calculate the maximum invariant subset of the approximate synchronization state set $\Lambda$.
(Please see Appendix for the proof of the validity of Algorithm 2.)
\begin{algorithm}
\caption{Calculate the maximum invariant subset $\mathcal{I}_m^\Lambda$.}
\begin{algorithmic}[1]
\Require  $\Lambda$, $L$
\Ensure $\mathcal{I}_m^\Lambda$
\State Initialize $\mathcal{I}_m^\Lambda=\Lambda$;
\While{$\delta_{k^{2n}}^i\in\Lambda$}
\If {$Col_i(L)\notin\Lambda$}
\State $\mathcal{I}_m^\Lambda=\mathcal{I}_m^\Lambda\setminus\{\delta_{k^{2n}}^i\}$;
\EndIf
\EndWhile
\end{algorithmic}
\end{algorithm}
\begin{remark}
The computational complexity of Algorithm 2 is $O(|\Lambda|k^{2n})$, $|\Lambda|\leq k^{2n}$.
Compared with the method proposed in reference \cite{guoyuqian2015Set}, its computational complexity is $O(|\Lambda|k^{6n})$ for calculating the maximum invariant subset of  $\Lambda$. Obviously, our approach can dramatically reduce computational complexity.
\end{remark}

Consider the approximately synchronous state set $\Lambda$ and a given initial state set $\Phi\subseteq\Delta_{k^{2n}}$. It is clear from Theorem \ref{th1}  that the key to approximate synchronization with respect to $\Phi$ is to guarantee that all attractors of $\Phi$ belong to $\Lambda$, i.e., $\Omega_\Phi\subseteq\Lambda$.
In fact, if $\Omega_\Phi\subseteq\Lambda$, then we can conclude that $\Omega_\Phi\subseteq \mathcal{I}_m^\Lambda$. Based on this, the following proposition is proposed  to illustrate the relationship between the MASB $\Phi_{\max}$ and $\Lambda$.
\begin{proposition}\label{pro3}
Consider the approximate synchronous state set $\Lambda$ and the maximum approximate synchronization basin $\Phi_{\max}$.
\begin{enumerate}[1)]
  \item $\Phi_{\max}=\emptyset$ if and only if $\mathcal{I}_m^\Lambda=\emptyset$.
  \item $\Phi_{\max}\neq\emptyset$ if and only if $\Omega_{\Phi_{\max}}\subseteq \mathcal{I}_m^\Lambda\neq\emptyset$.
\end{enumerate}
\end{proposition}

Next, we investigate the shortest approximate synchronization time.
Note that in vector form, $\Phi\sim\Psi$, so $\Gamma_{\Phi}=\Gamma_{\Psi}$.
In view of Theorem \ref{th0}, it is clear that $\Gamma_{\Phi}$ is the smallest integer $\rho$ such that $\xi(t;\xi_0)\in\Lambda$ holds for any $t\geq\rho$ and $\xi_0\in\Phi$.
Similar to \eqref{eq7} and \eqref{eq7-1}, denote the index vector of  $\mathcal{I}_m^\Lambda$ by
$\Xi_3=\sum_{\xi \in \mathcal{I}_m^\Lambda}\xi$.
Based on these, the following result can accurately  determine the shortest approximate synchronization time.
\begin{theorem}\label{pro4}
Consider  systems \eqref{sys3} and \eqref{sys4}. Suppose $\tau$ is the  transient period of system \eqref{sys5}.
\begin{enumerate}[1)]
  \item If systems \eqref{sys3} and \eqref{sys4} are approximately synchronous with respect to $\Phi\subseteq\Delta_{k^{2n}}$, then the shortest approximate  synchronization time with respect to $\Phi$ is
      \begin{equation}\label{eqpro4-1}
      \Gamma_{\Phi}=\mathop{\arg\min}_{1\leq t\leq \tau}\{sgn(L^t\Xi_1)\leq\Xi_3\}.
      \end{equation}

  \item If systems \eqref{sys3} and \eqref{sys4} are globally approximately synchronous, then the globally shortest approximate synchronization time is
      \begin{equation}\label{eqpro4-2}
        \Gamma=\mathop{\arg\min}_{1\leq t\leq \tau}\{sgn(L^t\mathbf{1}_{k^{2n}})\leq\Xi_3\}.
     \end{equation}
\end{enumerate}
\end{theorem}
\begin{proof}
See Appendix.
\end{proof}

\begin{remark}

In 1) of Theorem \ref{pro4}, the range of the shortest approximate synchronization time  with respect to $\Phi$ is limited to $1\leq\Gamma_{\Phi}\leq\tau$. In fact, the upper bound of $\Gamma_{\Phi}$ can be further reduced.
Define
\begin{equation}
\tau_\Phi:=\max_{\xi\in\Phi}\tau_\xi,
\end{equation}
where $\tau_\xi$ is the transient period of $\xi$. That is to say, $\tau_\Phi$ is the smallest integer such that $\xi(t;\xi_0)\in\Omega_\Phi$ holds for any $\xi_0\in\Phi$ and $t\geq\tau_\Phi$.
From the proof of Theorem \ref{pro4}, we can conclude that the shortest approximate synchronization time with respect to $\Phi$ can be also expressed as
\begin{equation}
\Gamma_{\Phi}=\mathop{\arg\min}_{1\leq t\leq \tau_\Phi}\{sgn(L^t\Xi_1)\leq\Xi_3\}.
\end{equation}
In conclusion, Theorem \ref{pro4} not only provides the method to calculate the shortest approximate synchronization time, but also indicates the relationship between transient period and the shortest approximate synchronization time.
\end{remark}

\subsection{Pinning Control Design for Approximate Synchronization}\label{sec4-3}

In this subsection, we investigate pinning control design for the approximate synchronization.
Precisely, if systems \eqref{sys1} and \eqref{sys2} are not approximately synchronous for some initial states, i.e., $\Phi_{\max}\neq\Delta_{k^{2n}}$, then we consider injecting pinning state feedback controllers to system \eqref{sys1} or \eqref{sys2} to guarantee that both of them achieve global approximate synchronization.

Different from [39], where drive-response BNs were considered and pinning controllers only inject to response system, this paper discusses coupled MVLNs, which can be unilaterally coupled or bidirectionally coupled. Therefore, how to select pinning nodes is a problem we need to deal with. The main idea in the following is to regard two coupled MVLNs as a whole and utilizing $\Phi_{\max}$ and $\mathcal{I}_m^\Lambda$ to find out pinning nodes.
Hence, for the convenience of  pinning control design, we denote
\begin{equation}
\xi_l(t)=
\begin{cases}
x_l(t), &1\leq l\leq n,\\
z_{l-n}(t), &n+1\leq l\leq 2n,
\end{cases}
\end{equation}
and
\begin{equation}
h_l(X(t),Z(t))=
\begin{cases}
f_l(X(t),Z(t)), &1\leq l\leq n,\\
g_{l-n}(X(t),Z(t)), &n+1\leq l\leq 2n,
\end{cases}
\end{equation}
where $\xi_l(t)\in\mathcal{D}_k$ and $h_l: \mathcal{D}_k^{2n}\rightarrow\mathcal{D}_k$, , $l=1,2,\ldots,2n$.
Denote the structure matrices of $h_l$ by $H_l$, $l=1,2,\ldots,2n$. Obviously, $H_l=F_l$, if $1\leq l\leq n$, otherwise, $H_l=G_{l-n}$.

Further, systems \eqref{sys1} and \eqref{sys2} with pinning controllers can be expressed in the form of
\begin{equation}\label{sys6}
  \begin{cases}
    \xi_i(t+1)=\overline{h}_i(u_i(t),X(t),Z(t)), ~i\in\mathcal{P},\\
    \xi_j(t+1)=h_j(X(t),Z(t)), ~j\in\{1,2,\ldots,2n\}\setminus\mathcal{P},
  \end{cases}
\end{equation}
where $\mathcal{P}$ is the set of pinning nodes,  $\overline{h}_i: \mathcal{D}_k^{2n+1}\rightarrow\mathcal{D}_k$, $i\in\mathcal{P}$ are logical functions, and $u_i(t):=\varphi_i(X(t),Z(t))\in\mathcal{D}_k$, $i\in\mathcal{P}$ are pinning state feedback controllers.
Moreover, $\overline{h}_i$ can be further expressed as
\begin{equation}
  \overline{h}_i(u_i(t),X(t),Z(t))=u_i(t)\odot_{\varphi_i}h_i(X(t),Z(t)),
\end{equation}
where $\odot_{\varphi_i}$, $i\in\mathcal{P}$ are some logical operators to be designed.

Assume $M_i$ and $K_i$ are the structure matrices of $\odot_{\varphi_i}$ and $\varphi_i$, $i\in\mathcal{P}$, respectively. Then system \eqref{sys6} can be converted into the following equivalent algebraic form:
\begin{equation}
\begin{cases}
  \xi_i(t+1)=\overline{H}_i\xi(t), ~&i\in\mathcal{P},\\
  \xi_j(t+1)=H_j\xi(t), ~&j\in\{1,2,\ldots,2n\}\setminus\mathcal{P},
\end{cases}
\end{equation}
where $\xi(t)=x(t)z(t)=\ltimes_{i=1}^{2n}\xi_i(t)\in\Delta_{k^{2n}}$, $\overline{H}_i=M_i(K_i\ast H_i)$.
Furthermore, we have
\begin{equation}\label{sys9}
\xi(t+1)=\overline{L}\xi(t),
\end{equation}
where $\overline{L}$ is the transition matrix with pinning controllers.

Obviously, $\mathcal{P}$ can be divided into a partition $\mathcal{P}=\mathcal{P}_1\cup\mathcal{P}_2$, $\mathcal{P}_1\cap\mathcal{P}_2=\emptyset$ such that for $i\in\mathcal{P}$,
\begin{equation}
\begin{cases}
i\in\mathcal{P}_1,~&1\leq i\leq n,\\
i-n\in\mathcal{P}_2,~&n+1\leq i \leq 2n.
\end{cases}
\end{equation}
Then it is clear that $\mathcal{P}_1$ and $\mathcal{P}_2$ correspond to the set of pinning nodes of systems \eqref{sys1} and \eqref{sys2}, respectively.
In addition, via injecting pinning state feedback controllers into the state nodes in $\mathcal{P}_1$ and $\mathcal{P}_2$, global approximate synchronization between systems \eqref{sys1} and \eqref{sys2} can be achieved.
If either $\mathcal{P}_1$ or $\mathcal{P}_2$ is an empty set, without loss of generality, suppose $\mathcal{P}_1\neq\emptyset$, but $\mathcal{P}_2=\emptyset$, then we can only add pinning controllers to system \eqref{sys1}.


In order to make systems \eqref{sys3} and \eqref{sys4} globally approximately synchronous, we try to perturb the transition matrix of system \eqref{sys5}.
Based on Proposition \ref{pro3}, the main process is divided into the following two situations.
\begin{enumerate}[1)]
\item Cases $1$: $\Phi_{\max}=\emptyset$.

It is noted that $\Phi_{\max}=\emptyset$ implies that $\mathcal{I}_m^\Lambda=\emptyset$.
In other words, $\Lambda$ contains no attractor of system \eqref{sys5}.
On the one hand, we need to change $Col_i(L)$ which is not in $\Lambda$ such that $Col_i(L)\in\Lambda$, $\delta_{k^{2n}}^i\in\Lambda$ to guarantee that  $\Lambda$ becomes an invariant subset of $\Delta_{k^{2n}}$.
On the other hand, we need to perturb all attractors of system \eqref{sys5} to make sure that the trajectory of system \eqref{sys5} starting from any initial state can go into $\Lambda$ after finite steps.
\item Case $2$: $\Phi_{\max}\neq\emptyset$.

If $\Phi_{\max}\neq\emptyset$, then $\mathcal{I}_m^\Lambda\neq\emptyset$ and $\Omega_{\Phi_{\max}}\subseteq \mathcal{I}_m^\Lambda$. That is to say, $\Lambda$ contains at least one attractor of system \eqref{sys5}.
Therefore, we only need to perturb the attractors that are not in $\Lambda$ to make sure that all states in $\Delta_{k^{2n}}$ can go into $\mathcal{I}_m^\Lambda$ after finite steps.
\end{enumerate}

For each attractor $\mathcal{C}_{i}\in\Omega$, $i\in\{1,2,\ldots,\epsilon\}$, we define the  index vector of $\mathcal{C}_{i}$ as:
\begin{equation}\label{eq7-2}
\Xi^{i}=\sum_{\delta_{k^{2n}}^i\in\mathcal{C}_{i}} \delta_{k^{2n}}^i\in\mathcal{B}_{k^{2n}\times1}.
\end{equation}
Denote the set of attractors which need to be perturbed by $\Omega_0$.
An algorithm (Algorithm 3) to calculate $\Omega_0$ is derived. (Please see Appendix for the proof of the validity of Algorithm 3.)

\begin{algorithm}[htb]
\label{alg1}
\caption{Calculate the set $\Omega_0$ containing all attractors needed to be perturbed.}
\begin{algorithmic}[1]
\Require $\Omega=\{\mathcal{C}_{1}, \mathcal{C}_{2},\ldots,\mathcal{C}_{\epsilon}\}$, $\Lambda$
\Ensure $\Omega_0$
\State Initialize $\Omega_0=\{\mathcal{C}_{1}, \mathcal{C}_{2},\ldots,\mathcal{C}_{\epsilon}\}$;
\State Compute $\Xi_2$ according to \eqref{eq7-1};
\For{$i=0 \to \epsilon$}
\State Compute $\Xi^{i}$ according to \eqref{eq7-2};
\If{$\Xi^{i}\circ\Xi_2=\Xi^{i}$}
\State $\Omega_0=\Omega_0\backslash\mathcal{C}_{i}$;
\EndIf
\EndFor
\end{algorithmic}
\end{algorithm}

Based on $\Phi_{\max}$, $\mathcal{I}_m^\Lambda$ and $\Omega_0$ obtained in Algorithms 1-3, the following algorithm (Algorithm 4) is given to determine the perturbed transition matrix and the set of pinning nodes.
\begin{algorithm}\label{alg2}
\caption{Calculate the perturbed transition matrix $\overline{L}$ and set $\mathcal{P}$ containing pinning nodes.}
\begin{algorithmic}[1]
\Require  $\Lambda$, $\mathcal{I}_m^\Lambda$, $\Phi_{\max}$, $\Omega_0:=\{\mathcal{C}^1,\ldots,\mathcal{C}^{|\Omega_0|}\}$, $L$ and $H_i$, $i=1,2,\ldots,2n$
\Ensure $\overline{L}$, $\mathcal{P}$ and $\overline{H}_i$, $i\in\mathcal{P}$
\State Initialize $\mathcal{P}=\emptyset$ and $\overline{L}=L$;
\If{$\Phi_{\max}=\emptyset$}
\While{$\delta_{k^{2n}}^i\in\Lambda$}
\If{$Col_i({L})\in\Lambda$}
\State Let $Col_i({\overline{L}})=Col_i({L})$;
\Else
\State Find a state $\xi_1\in\Lambda$ and let $Col_i({\overline{L}})=\xi_1$;
\EndIf
\EndWhile
\For{$i=1 \to |\Omega_0|$}
\State Select randomly a state $\delta_{k^{2n}}^l\in\mathcal{C}^i$;
\State Find a state $\xi_2\in \Lambda$ and let $Col_l({\overline{L}})=\xi_2$;
\EndFor
\Else
\For{$i=1 \to |\Omega_0|$}
\State Select randomly a state $\delta_{k^{2n}}^l\in\mathcal{C}^i$;
\State Find a state $\xi_3\in \mathcal{I}_m^\Lambda$ and let $Col_l({\overline{L}})=\xi_3$;
\EndFor
\EndIf
\For{$i=1 \to 2n$}
\State Compute $\overline{H}_i:=(\mathbf{1}^\top_{k^{i-1}}\otimes I_k\otimes\mathbf{1}^\top_{k^{2n-i}})\overline{L}$;
\If{$\overline{H}_i\neq H_i$}
\State $\mathcal{P}=\mathcal{P}\cup\{i\}$;
\EndIf
\EndFor
\end{algorithmic}
\end{algorithm}

\begin{remark}
 (A note on the selection of pinning nodes)

The first step is to perturb the transition matrix L of system (14). In Algorithm 4, lines 2-13 deal with the case of $\Phi_{\max}=\emptyset$. Lines 3-9
make sure that the approximate synchronous state set $\Lambda$ becomes an invariant subset of $\Delta_{k^{2n}}$. And lines 10-13 guarantee that all states will enter into $\Lambda$ after finite steps.
To deal with the case of $\Phi_{\max}\neq\emptyset$, lines 15-18 perturb the attractors which are not in $\Lambda$.
The second step is to find out the pinning nodes from $\overline{L}$. lines 20-25 aim to determine pinning nodes by comparing the structure matrices $\overline{H}_i$ with $H_i$.
\end{remark}

The following  theorem is given to demonstrate that systems \eqref{sys3} and \eqref{sys4} are  globally approximately synchronous through the process in Algorithm 4.
\begin{theorem}\label{th3}
Consider systems \eqref{sys3} and \eqref{sys4} with the augmented system \eqref{sys5}.  If the transition matrix $L$ of system (14) can be changed into $\overline{L}$ by Algorithm 4, then systems \eqref{sys3} and \eqref{sys4} are globally approximately synchronous.
\end{theorem}
\begin{proof}
See Appendix.
\end{proof}

Next, we concern with how to design feasible pinning state feedback controllers according to the perturbed transition matrix $\overline{L}$.

Note that by Algorithm 4, $\overline{H}_i$, $i\in\mathcal{P}$ can be calculated from $\overline{L}$ . The key process we need to deal with is to solve $M_i$ and $K_i$ from the following logical matrix equation:
\begin{equation}\label{eq5}
 \overline{H}_i=M_i(K_i\ast H_i),~i\in\mathcal{P}.
\end{equation}
Actually, equation \eqref{eq5} is solvable.
In light of the property of  Khatri-Rao product, \eqref{eq5} is equivalent to
\begin{equation}\label{eq8}
Col_l(\overline{H}_i)=M_iCol_l(K_i)Col_l(H_i),~l=1,2,\ldots,k^{2n}.
\end{equation}
Assume
\begin{align}
H_i=&\delta_{k}[h^i_1, h^i_2,\ldots,h^i_{k^{2n}}],\notag\\
\overline{H}_i=&\delta_{k}[\overline{h}^i_1, \overline{h}^i_2,\ldots,\overline{h}^i_{k^{2n}}],\notag\\
K_i=&\delta_{k}[w^i_1, w^i_2,\ldots,w^i_{k^{2n}}],\notag\\
M_i=&\delta_{k}[m^i_1, m^i_2,\ldots,m^i_{k^{2}}].
\end{align}
Then from  \eqref{eq8}, it is easy to calculate
\begin{align}
  \delta_k^{\overline{h}^i_l}&=M_i\delta_k^{w^i_l}\delta_k^{h^i_l}=M_i\delta_{k^2}^{(w^i_l-1)k+h^i_l}=\delta_k^{m^i_{(w^i_l-1)k+h^i_l}}.
\end{align}
Hence,
\begin{equation}\label{eq6}
  m^i_{(w^i_l-1)k+h^i_l}=\overline{h}^i_l, l=1,2,\ldots,k^{2n}.
\end{equation}
According to equation \eqref{eq6}, as long as $K_i$ is known, then $M_i$ can be determined subsequently.
Nevertheless, it is worth noting that the value of $w^i_l$ is not arbitrary, because it needs to make sure that the value of $m^i_{(w^i_l-1)k+h^i_l}$ is unique.
For instance, suppose $h^i_1=1, \overline{h}^i_1=1$ and $h^i_2=1, \overline{h}^i_2=2$, it follows from \eqref{eq6} that $w^i_1\neq w^i_2$. Otherwise, if $w^i_1=w^i_2=1$, then we can see $1=\overline{h}^i_1=m^i_1=\overline{h}^i_2=2$, which is a contradiction.

In order to ensure that the value of $m^i_{(w^i_l-1)k+h^i_l}$, $l\in\{1,2,\ldots, k^2\}$ makes sense, it is necessary to give a criterion for the value of $w^i_l$, $l\in\{1,2,\ldots,k^{2n}\}$. Then construct
\begin{equation}\label{eq9}
T_{\alpha,\beta}^i=\{l: h^i_l=\alpha, \overline{h}^i_l=\beta\},~\alpha,\beta=1, 2,\ldots,k.
\end{equation}
It is clear  that
\begin{equation}
\bigcup_{\alpha=1}^k\bigcup_{\beta=1}^kT_{\alpha,\beta}^i=\{1,2,\ldots,k^{2n}\},
\end{equation}
where $T_{\alpha,\beta}^i$, $\alpha,\beta=1, 2,\ldots,k$ are pairwise disjoint.
Denote
\begin{equation}\label{eq10}
\mathfrak{T}_\sigma^i=\{T_{\sigma,\beta}^i: T_{\sigma,\beta}^i\neq\emptyset,\beta\in\{1, 2,\ldots,k\}\},~\sigma=1, 2,\ldots,k.
\end{equation}
Obviously, $\mathfrak{T}_\sigma^i$ consists of all nonempty set  $T_{\sigma,\beta}^i$, $\beta\in$ $\{1,$ $2,$ $\ldots,k\}$.

On the basis of the analysis above, we present the following  proposition to explain  the solvability of equation \eqref{eq6}.
\begin{proposition}\label{pro2}
Consider equation \eqref{eq6}. The following conclusions provide the criteria for the value of $w^i_l$.
\begin{enumerate}[1)]
\item If $|\mathfrak{T}_\sigma^i|\leq1$, $\sigma=1,2,\ldots,k$, then $w^i_l$, $l=1,2,\ldots,k^{2n}$ can take any value from $\{1, 2,\dots,k\}$.
\item If there exists $\sigma^\prime\in\{1,2,\ldots,k\}$ such that $2\leq|\mathfrak{T}_{\sigma\prime}^i|\leq k$, then equation $\eqref{eq6}$ makes sense  if and only if  for any $l_{\mu}\in T_{\sigma^\prime,\beta_\mu}^i$ and $l_{\nu}\in T_{\sigma^\prime,\beta_\nu}^i$,
   \begin{equation}
    w^i_{l_\mu}\neq w^i_{l_\nu},
   \end{equation}
where $T_{\sigma^\prime,\beta_\mu}^i, T_{\sigma^\prime,\beta_\nu}^i\in\mathfrak{T}_{\sigma^\prime}^i$ and $\mu\neq\nu$.
\end{enumerate}
\end{proposition}
\begin{proof}
See Appendix.
\end{proof}

Combined with Proposition \ref{pro2}, an algorithm (Algorithm 5) is developed to design the structure matrices $K_i$ and $M_i$, $i\in\mathcal{P}$.
Note that  $w^i_l$ and $m^i_{(w^i_l-1)k+h^i_l}$ may have multiple values satisfying  condition \eqref{eq6}. We only provide one feasible method in Algorithm 5.
\begin{algorithm}[htbp]\label{alg2}
\caption{Calculate structure matrices $K_i$ and $M_i$, $i\in\mathcal{P}$.}
\begin{algorithmic}[1]
\Require $H_i$ and $\overline{H}_i$, $i\in\mathcal{P}$
\Ensure $K_i$ and $M_i$, $i\in\mathcal{P}$
\State Denote $H_i$ $=$ $\delta_{k}[h^i_1,h^i_2,\ldots,h^i_{k^{2n}}]$, $\overline{H}_i$ $=$ $\delta_{k}[\overline{h}^i_1,$ $\overline{h}^i_2,$ $\ldots,$ $\overline{h}^i_{k^{2n}}]$, $K_i$ $=$ $\delta_{k}[w^i_1,$ $w^i_2,$ $\ldots,$ $w^i_{k^{2n}}]$, $M_i$ $=$ $\delta_{k}[m^i_1,$ $m^i_2,$ $\ldots,$ $m^i_{k^{2}}]$;
\While{$i\in\mathcal{P}$}
\State Compute $T_{\alpha,\beta}^i$, $\alpha,\beta=1,2,\ldots,k$ according to \eqref{eq9};
\State Compute $\mathfrak{T}_\sigma^i$, $\sigma=1,2,\ldots,k$ according to \eqref{eq10};
\State Compute $|\mathfrak{T}_\sigma^i|$, $\sigma=1,2,\ldots,k$;
\State Denote $\mathfrak{T}_\sigma^i$ $=$ $\{T_{\sigma,\beta_1}^i,T_{\sigma,\beta_2}^i,\ldots,
T_{\sigma,\beta_{|\mathfrak{T}_\sigma^i|}}^i\}$, $\sigma=1,2,$ $\ldots,$ $k$;
\For{$\sigma=1\to k$}
\If{$1\leq |\mathfrak{T}_\sigma^i| \leq k$}
\For{$\kappa=1\to|\mathfrak{T}_\sigma^i|$}
\State For all $l\in T_{\sigma,\beta_\kappa}^i$, let $w^i_l=\kappa$;
\EndFor
\EndIf
\EndFor
\For{$l=1 \to k^{2n}$}
\State Compute $\pi:=(w^i_l-1)k+h^i_l$;
\State Let $m^i_\pi=\overline{h}^i_l$;
\EndFor
\EndWhile
\end{algorithmic}
\end{algorithm}

\begin{remark}
Analysis about computational complexity of Algorithms 3-5.
\begin{enumerate}[1)]
  \item Note that in Algorithm 3, the computational complexity of line 2 is $O(k^{2n})$ and the complexity of lines 3-10 is $O(k^{2n})$. Hence, the computational complexity of Algorithm 2 is $O(k^{2n})$.
  \item In Algorithm 4,  the computational complexity of lines 3-9 is $O(|\Lambda|)$, $|\Lambda|\leq k^{2n}$. The complexity of both ``for" loops in lines 10-13 and 15-18 is $O(|\Omega_0|)$, $|\Omega_0|\leq k^{2n}$. Besides, the complexity of lines 20-25 is $O(2nk^{4n+1})$. Therefore, we conclude that the computational complexity of Algorithm 4 is $O(2nk^{4n+1}+|\Lambda|+|\Omega_0|)=O(2nk^{4n+1})$.
  \item From Algorithm 5, the computational complexity of lines 3-6 is $O(k^{2n})$. And the computational complexity of  two ``for" loops in lines 7-18 are $O(k^{2n})$. Hence, the computational complexity of Algorithm 5 is $O(k^{2n})$.
\end{enumerate}
\end{remark}

\begin{remark}\label{rem2}
Review the shortest approximate synchronization time discussed in Theorem \ref{pro4}.
For systems \eqref{sys1} and \eqref{sys2} with pinning controllers,  we can also consider the problem of  global shortest approximate synchronization time.
By perturbing  transition matrix $L$ in Algorithm 4,
system \eqref{sys5} changes into
\begin{equation}\label{sys10}
\xi(t+1)=\overline{L}\xi(t).
\end{equation}
It is worth pointing out that the pinning controllers are not unique because $\overline{L}$, depending on lines 2-19 in Algorithm 4, is not unique.
However, it is clear from Theorem \ref{pro4} that once $\overline{L}$ is determined,  the global shortest approximate synchronization time of system \eqref{sys5} is
\begin{align}
 \Gamma=\mathop{\arg\min}_{1\leq t\leq \overline{\tau}}\{sgn(\overline{L}^t\mathbf{1}_{k^{2n}})\leq\Xi_3\},
\end{align}
where $\overline{\tau}$ is the transient period of system \eqref{sys10}.
Hence, it is an interesting issue to further investigate how to design pinning controllers to make the global shortest approximate synchronization time as small as possible in the future research work.
\end{remark}

\section{Special Case of Approximate Synchronization}\label{sec6}

In Sections \ref{sec3} and \ref{sec4}, we have investigated the approximate synchronization  of two coupled $k$-valued ($k>2$) logical networks.
As a matter of fact, the relevant results in this paper can be strengthened and applied to the complete synchronization problem of $k$-valued ($k\geq2$) logical networks.

Reconsider Definition \ref{def1}. As described in Remark \ref{rem1}, if take $\gamma=0$  in condition \eqref{eqrem1}, i.e.,
\begin{equation}
\max_{1\leq i\leq n}|x_i(t;X_0,Z_0)-z_i(t;X_0,Z_0)|=0,
\end{equation}
which means that the error between corresponding nodes is zero,
then approximate synchronization becomes complete synchronization.
In this case, approximate synchronous state set $\Lambda$ becomes synchronous state set $\Lambda^\prime$, where
\begin{equation}\label{eq11}
\Lambda^\prime=\{\delta_{k^{2n}}^{(i-1)k^n+j}: i=j, i=1,2,\ldots,k^n\},
\end{equation}
and $|\Lambda|=k^n$.
Meanwhile, the maximum approximate synchronization basin $\Phi_{\max}$ turns into the maximum synchronization basin $\Phi^\prime_{\max}$ and the shortest approximate synchronization time $\Gamma_{\Phi}$ becomes the shortest synchronization time $\Gamma^\prime_{\Phi^\prime_{\max}}$.

It is worth noting that if $\Lambda^\prime$ is substituted for $\Lambda$, then all theoretical results about approximate synchronization obtained in this paper are applicable to solve the complete synchronization problem of $k$-valued ($k\geq2$) logical networks. In Section \ref{sec5}, two examples (Examples \ref{ex1} and \ref{ex4}) are given to further illustrate the statement above.

\begin{remark}
Some advantages of this paper are summarized as follows:

1)  Global and local approximate synchronization proposed in this paper are two novel concept in the fields of logical networks. Compared with \cite{lirui2012CompleteSynchronization, mengmin2014multi-valued, chenhongwei2020LocalSynchronization}, our results are more general and can degenerate into the complete synchronous situation.
In addition, in this paper, the shortest approximate synchronization time which has a concrete physical meaning, are introduced and its calculation method is presented. However, the shortest synchronization time was not specified in \cite{lirui2012CompleteSynchronization, mengmin2014multi-valued, chenhongwei2020LocalSynchronization}.

2) The method of pinning control strategy used in this paper improves and extends
the approach in reference [39]. First, in [39], pinning controllers were imposed on response system by regarding the states of drive system as switching signals, which is not applicable to general coupled MVLNs. In this paper, we regard two coupled MVLNs as a whole and take advantage of the maximum approximate synchronization basin and the maximum invariant subset to determine pinning nodes. Second, in the process of calculating the structure matrices of pinning state feedback controllers, we give explicit criteria for the solvability of the structure matrices. However, in reference [39], the discriminant criteria were not given, and the structure matrices were obtained through solving logic matrix equations based on enumeration method, which is rather tedious.
\end{remark}

\section{Examples}\label{sec5}

In this section, four examples are given to illustrate the effectiveness of the results obtained in Section \ref{sec4}. Examples \ref{ex2} and \ref{ex3} discuss the approximate synchronization of two bidirectionally  and unidirectionally coupled logical networks, respectively. Examples  \ref{ex1} and \ref{ex4} analyze the complete synchronization of nonlinear feedback shift registers (NFSRs) and an epigenetic model, respectively.

\begin{example}\label{ex2}
Consider the following two bidirectionally coupled 5-valued logical networks:
\begin{equation}\label{sys7}
\begin{cases}
x_1(t+1)=x_1(t)\vee\neg x_3(t)\vee z_1(t),\\
x_2(t+1)=x_1(t)\vee\neg x_2(t),\\
x_3(t+1)=z_2(t),
\end{cases}
\end{equation}
and
\begin{equation}\label{sys8}
\begin{cases}
z_1(t+1)=x_1(t)\vee z_1(t)\vee \neg z_3(t),\\
z_2(t+1)=z_1(t)\vee\neg z_2(t),\\
z_3(t+1)=\oslash_5(x_2(t)),
\end{cases}
\end{equation}
where $x_i(t)\in\mathcal{D}_5$, $z_i(t)\in\mathcal{D}_5$, $i=1,2,3$, are state variables of systems \eqref{sys7} and \eqref{sys8}, respectively.
Denoting $x(t)=\ltimes_{i=1}^3x_i(t)\in\Delta_{5^3}$, $z(t)=\ltimes_{i=1}^3z_i(t)\in\Delta_{5^3}$ and $\xi(t)=x(t)z(t)\in\Delta_{5^6}$, we can convert systems \eqref{sys7} and \eqref{sys8} into the algebraic form \eqref{sys3} and \eqref{sys4} with
$F\in\mathcal{L}_{5^3\times5^6}$, $G\in\mathcal{L}_{5^3\times5^6}$.
Then the augmented system of  \eqref{sys7} and \eqref{sys8} can be expressed as
\begin{equation}\label{ex-eq1}
\xi(t+1)=L\xi(t),
\end{equation}
where $L=\delta_{5^6}$ $[2~2~2~2~2~127~\cdots~576~551~526~501]\in\mathcal{L}_{5^6\times5^6}$.

Let
\begin{equation}\label{eq-3}
\begin{cases}
i=\Sigma_{l=1}^2(i_{3-l}-1)5^l+i_3,\\
j=\Sigma_{l=1}^2(j_{3-l}-1)5^l+j_3,
\end{cases}
\end{equation}
where $|i_\varepsilon-j_\varepsilon|\leq1$, $i_\varepsilon, j_\varepsilon\in\{1,2,3,4,5\}$, $\varepsilon=1,2,3$.
According to \eqref{eq-1}$-$\eqref{eq2}, the approximate synchronous state set is
\begin{align}\label{eqex1Lamda}
\Lambda=&\{\delta_{5^6}^{(i-1)5^3+j}: i, j ~\text{satisfy}~\eqref{eq-3}\},
\end{align}
where $|\Lambda|=2197$, which is consistent with Proposition \ref{pro1}.

A simple calculation shows that the set of all attractors of system \eqref{ex-eq1} is $\Omega=\{\mathcal{C}_1, \mathcal{C}_2\}$, where $\mathcal{C}_1=\{\delta_{5^6}^{3}\}$, $\mathcal{C}_2=\{\delta_{5^6}^{3908}\}$ and the transient period of system \eqref{ex-eq1} is $\tau=9$.
In light of Corollary \ref{th2}, it is easy to get that
\begin{equation}
sgn(L^9\mathbf{1}_{5^6})\leq\Xi_2,
\end{equation}
where $\Xi_2=\sum_{\delta_{5^6}^i\in\Lambda}\delta_{5^6}^i$.
Therefore, systems \eqref{sys7} and \eqref{sys8} are globally approximately synchronous.
Besides, using Algorithm 1, we can get that $\Phi_{\max}=\Delta_{k^{2n}}$, which also implies that systems \eqref{sys7} and \eqref{sys8} can achieve global approximate synchronization.

According to Theorem \ref{pro4}, the global shortest approximate synchronization time is $\Gamma=9$.
Take initial states $X_0=(\frac{1}{4}, 1, 1)$ and $Z_0=(\frac{1}{4}, \frac{1}{4}, \frac{1}{4})$. It is clear that $(X_0, Z_0)\sim\delta_{5^6}^{9469}$ and
the shortest approximate synchronization time with respect to $\Phi_0:=\{\delta_{5^6}^{9469}\}$ is $\Gamma_{\Phi_0}=6$. The trajectories of systems \eqref{sys7} and \eqref{sys8} starting from $(X_0, Z_0)$ are shown in Fig \ref{fig2}.
\begin{figure}[htb]
\centering
\includegraphics[width=3.3in]{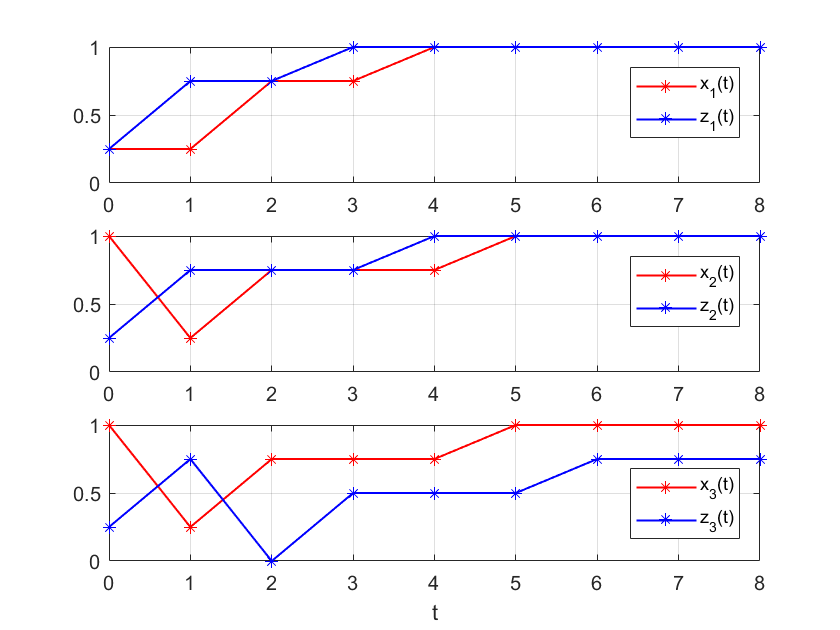}
\caption{ The evolution of systems \eqref{sys7} and \eqref{sys8} starting from initial states $X_0=(\frac{1}{4}, 1, 1)$ and $Z_0=(\frac{1}{4}, \frac{1}{4}, \frac{1}{4})$.}
\label{fig2}
\end{figure}

\end{example}

\begin{example}\label{ex3}
Consider the following two unidirectionally coupled 5-valued logical networks:
\begin{equation}\label{ex2-eq1}
\begin{cases}
x_1(t+1)=x_3(t),\\
x_2(t+1)=\neg x_1(t)\wedge x_2(t) \wedge x_3(t),\\
x_3(t+1)=x_2(t),
\end{cases}
\end{equation}
and
\begin{equation}\label{ex2-eq2}
\begin{cases}
z_1(t+1)=z_3(t),\\
z_2(t+1)=(\neg z_1(t) \wedge z_2(t) \wedge z_3(t))\vee x_1(t) ,\\
z_3(t+1)=z_2(t)\vee x_2(t).
\end{cases}
\end{equation}
Analogous to the analysis in Example \ref{ex2},  we can get the augmented system as follows:
\begin{equation}\label{ex2-eq3}
\xi(t+1)=L\xi(t),
\end{equation}
where $\xi(t)\in\Delta_{5^6}$ and
$L=\delta_{5^6}$ $[2501~2526~$ $2551~$ $\cdots~$ $15575~$ $15600$ $15625]$ $\in$ $\mathcal{L}_{5^6\times5^6}$.
And the approximate synchronous state set is consistent with \eqref{eqex1Lamda}.

A simple calculation shows that system \eqref{ex2-eq3} has six attractors. That is $\Omega=\{\mathcal{C}_1, \mathcal{C}_2, \mathcal{C}_3, \mathcal{C}_4, \mathcal{C}_5, \mathcal{C}_6\}$, where $\mathcal{C}_1=\{\delta_{5^6}^{7813}\}$, $\mathcal{C}_2=\{\delta_{5^6}^{11688}\}$, $\mathcal{C}_3=\{\delta_{5^6}^{11719}\}$, $\mathcal{C}_4=\{\delta_{5^6}^{15563}\}$, $\mathcal{C}_5=\{\delta_{5^6}^{15594}\}$ and $\mathcal{C}_6=\{\delta_{5^6}^{15625}\}$.
And the transient period of system $\eqref{ex2-eq3}$ is $\tau=8$.

In light of Algorithm 1,  we can derive that the maximum approximate synchronization basin satisfying $|\Phi_{\max}|=15376<5^6$.
That is to say, systems \eqref{ex2-eq1} and \eqref{ex2-eq2} are approximately synchronous with respect to $\Phi_{\max}$. From Theorem \ref{pro4}, we derive that the shortest synchronization time with respect to $\Phi_{\max}$ is $\Gamma_{\Phi_{\max}}=8$.
Take $\xi_0=\delta_{5^6}^{11755}\in\Phi_{\max}$. Note that $\delta_{5^6}^{11755}\sim(\frac{1}{4},\frac{1}{4},0,1,1,0)$. Then the trajectories of systems \eqref{ex2-eq1} and \eqref{ex2-eq2} starting from $X_0=(\frac{1}{4},\frac{1}{4},0)$ and $Z_0=(1,1,0)$ are shown in Fig. \ref{fig4}.
\begin{figure}[htb]
\centering
\includegraphics[width=3.3in]{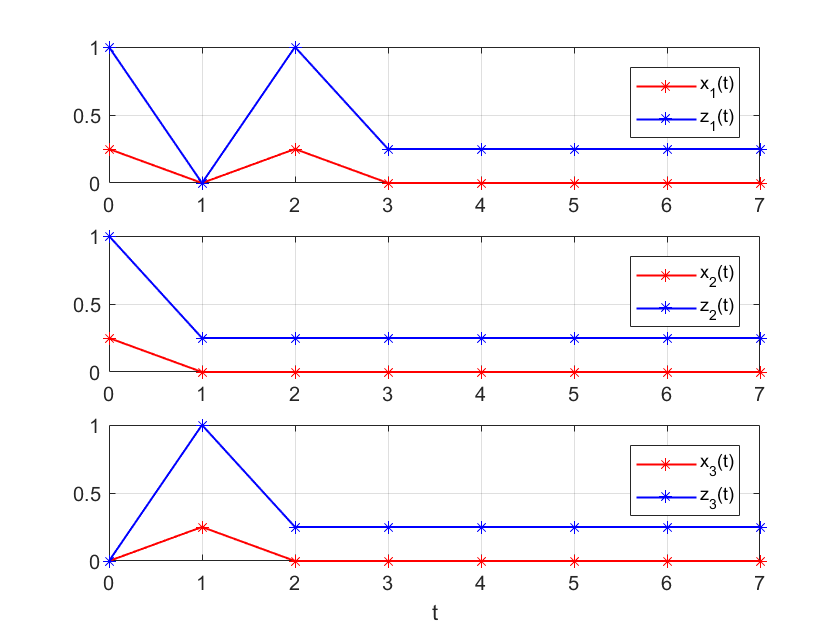}
\caption{The evolution of systems \eqref{ex2-eq1} and \eqref{ex2-eq2} starting from initial states $X_0=(\frac{1}{4},\frac{1}{4},0)$ and $Z_0=(1, 1, 0)$.}
\label{fig4}
\end{figure}

Since $\Phi_{\max}\neq\Delta_{5^6}$, systems \eqref{ex2-eq1} and \eqref{ex2-eq2} are not globally approximately synchronous according to Corollary \ref{cor1}.
Now, we try to draw pinning controllers into systems \eqref{ex2-eq1} and \eqref{ex2-eq2} to make the two systems achieve global approximate synchronization.
According to Algorithm 3, we can get that the set of attractors needed to be perturbed is $\Omega_0=\{\mathcal{C}_4\}$, where $\mathcal{C}_4=\{\delta_{5^6}^{15563}\}$.
Based on Algorithm 4, we derive that the transition matrix $L$ can be perturbed into
\begin{align*}
\overline{L}=\delta_{5^6}[2501~2526~2551\cdots15575~15600~15625]\in\mathcal{L}_{5^6\times5^6},
\end{align*}
and the set of pinning nodes is $\mathcal{P}=\{2,3,4\}$.
That is to say, $\mathcal{P}_1=\{2,3\}$ and $\mathcal{P}_2=\{1\}$.
From Theorem \ref{th3},
if pinning controllers are injected into the second and third nodes of system \eqref{ex2-eq1} and the first node of system \eqref{ex2-eq2}, then systems \eqref{ex2-eq1} and \eqref{ex2-eq2} achieve  global approximate synchronization.
Furthermore, the global shortest synchronization time is $\Gamma=9$  by Theorem \ref{pro4}.

Assume that the pinning state feedback controllers are $u_i(t)=\varphi_i(X(t),Z(t))$, $i\in\mathcal{P}$.
Through Algorithm 5, we can get that the structure matrices of $\varphi_i$, $i\in\mathcal{P}$ are
\begin{align*}
K_2=\delta_{5}[5~5~5~5~5~5~5~5~\cdots~5~5~5~5~5~5~5~5]\in\mathcal{L}_{5\times5^6},\\
K_3=\delta_{5}[1~1~1~1~1~1~1~1~\cdots~5~5~5~5~5~5~5~5]\in\mathcal{L}_{5\times5^6},\\
K_4=\delta_{5}[1~2~3~4~5~1~2~3~\cdots~ 3~ 4~ 5~ 1~2~ 3~4 ~5]\in\mathcal{L}_{5\times5^6}.
\end{align*}
Using Lemma \ref{lem3}, one can obtain the logical expressions of pinning state feedback controllers $\varphi_i$, $i\in\mathcal{P}$ recursively. Due to the limitation of space, we omit it here.
In addition, the the structure matrices of logical operators $\odot_{\varphi_i}$, $i\in\mathcal{P}$ are
\begin{align*}
M_2=\delta_5[1~1~1~1~1 ~1~2~2~2~2 ~1~2~3~3~3 ~1~2~3~4~4~ 1~2~3~4~5],\\
M_3=\delta_5[1~1~1~1~1 ~1~2~2~2~2 ~1~2~3~3~3 ~1~2~3~4~4~ 1~2~3~4~5],\\
M_4=\delta_5[1~2~3~4~5 ~2~2~3~4~5 ~3~3~3~4~5 ~4~4~4~4~5~ 5~5~5~5~5].
\end{align*}
It is easy to know that $\odot_{\varphi_2}=\vee$, $\odot_{\varphi_3}=\vee$ and $\odot_{\varphi_4}=\wedge$.

Take initial states $X_0=(\frac{2}{4}, \frac{2}{4}, 0)$ and $Z_0=(1, 1, \frac{2}{4})$. It is not hard to check that
\begin{equation*}
 (X_0,Z_0)=(\frac{2}{4}, \frac{2}{4}, 0, 1, 1, \frac{2}{4})\sim\delta_{5^6}^{8003}\notin\Phi_{max},
\end{equation*}
which means systems \eqref{ex2-eq1} and \eqref{ex2-eq2} cannot achieve approximate synchronization with respect to $\Phi_1:=\{\delta_{5^6}^{8003}\}$.
Fig. \ref{fig3} shows that the trajectories of systems \eqref{ex2-eq1} and \eqref{ex2-eq2} staring from $(X_0,Z_0)$ will enter into $(0, 0, 0, \frac{2}{4}, \frac{2}{4}, \frac{2}{4})$ after $3$ steps.  Note that $(0, 0, 0, \frac{2}{4}, \frac{2}{4}, \frac{2}{4})\sim \delta_{5^6}^{15563}\notin\Phi_{max}$, which is an attractor of system \eqref{ex2-eq3}.
If pinning controllers are drawn into systems \eqref{ex2-eq1} and \eqref{ex2-eq2} from
$t=4$, then by Theorem \ref{pro4}, we can get that the shortest approximate synchronization time with respect to $\Phi_2:=\{\delta_{5^6}^{15563}\}$ is $\Gamma_{\Phi_2}=6$, which is consistent with Fig. \ref{fig3}.
\begin{figure}[htb]
\includegraphics[width=5.5in]{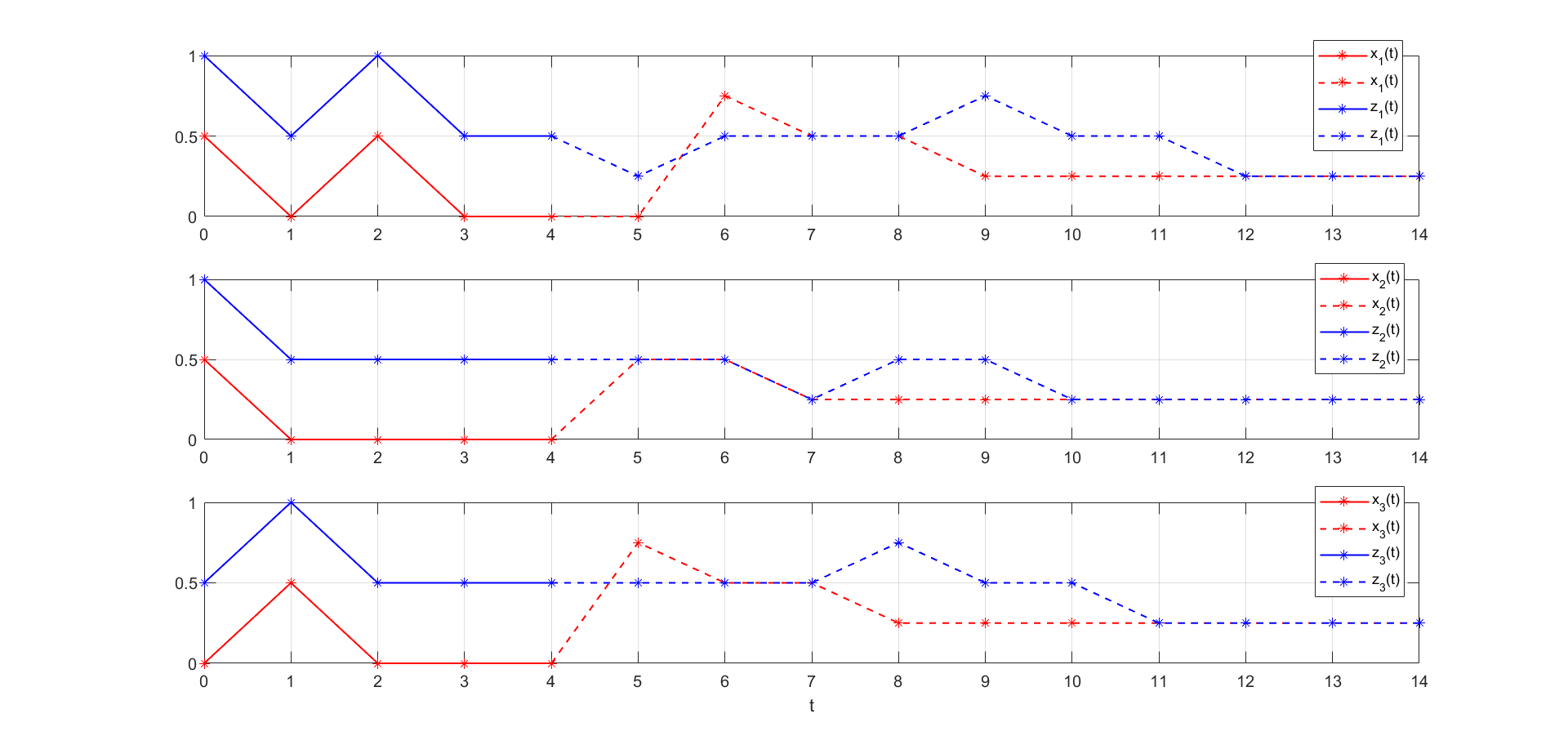}
\caption{Initial states: $X_0=(\frac{2}{4}, \frac{2}{4}, 0)$ and $Z_0=(1, 1, \frac{2}{4})$.
Solid lines: the states of the original systems \eqref{ex2-eq1} and \eqref{ex2-eq2}, dotted lines: the states of \eqref{ex2-eq1} and \eqref{ex2-eq2} with pinning controllers.}
\label{fig3}
\end{figure}
\end{example}

\begin{example}\label{ex1}
NFSRs are widely used in the field of information security, such as stream cipher and convolutional decoder.
A cascade connection of an NFSR into another NFSR is the main components in the Grain family of stream ciphers \cite{zhongjianghua2019NFSR}. Its structure diagram is shown in Fig. \ref{fig1}.
In this example, we consider a 4-valued cascade connection of two 3-stage NFSRs. Its working mechanism can be expressed as the following two coupled systems:
\begin{equation}\label{ex3-eq1}
\begin{cases}
x_1(t)=x_2(t),\\
x_2(t)=x_3(t),\\
x_3(t)=z_1(t)\oplus_4 f(x_1(t),x_2(t),x_3(t)),
\end{cases}
\end{equation}
\begin{equation}\label{ex3-eq2}
\begin{cases}
z_1(t)=z_2(t),\\
z_2(t)=z_3(t),\\
z_3(t)=g(z_1(t),z_2(t),z_3(t)),
\end{cases}
\end{equation}
where $x_i(t)\in\mathcal{D}_4$, $z_i(t)\in\mathcal{D}_4$, $i=1,2,3$, are state variables of NFSR $1$ and NFSR $2$, respectively;
$f(x_1(t),x_2(t),x_3(t))=x_1(t)\wedge x_2(t)\wedge\neg x_3(t)$, $g(z_1(t),z_2(t),z_3(t))=z_1(t)\wedge z_2(t)\wedge\neg z_3(t)$ are nonlinear feedback functions of NFSR $1$ and NFSR $2$, respectively.
\begin{figure}[H]
\centering
\includegraphics[width=3.5in]{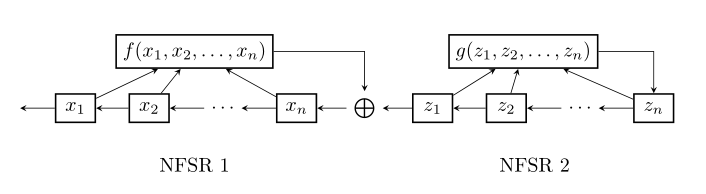}
\caption{Structure diagram of a cascade connection of two NFSRs }
\label{fig1}
\end{figure}

Let $x(t)=\ltimes_{i=1}^3x_i(t)\in\Delta_{4^3}$, $z(t)=\ltimes_{i=1}^3z_i(t)\in\Delta_{4^3}$ and $\xi(t)=x(t)z(t)\in\Delta_{4^6}$. Systems \eqref{ex3-eq1} and \eqref{ex3-eq2} can be converted into the following augmented system:
\begin{equation}\label{ex3-eq3}
\xi(t+1)=L\xi(t),
\end{equation}
where
$L=\delta_{4^6}$ $[4~7~10~13$ $\cdots$ $4084~4088$ $4092~4096]$ $\in$ $\mathcal{L}_{4^6\times4^6}$.
According to \eqref{eq11}, the synchronous state set of system \eqref{ex3-eq3} is
$\Lambda^\prime=\{\delta_{4^6}^{(i-1)4^3+i}: i=1,2,\ldots, 4^3\}\notag$
with $|\Lambda^\prime|=64$.

A simple calculation shows that system \eqref{ex3-eq3} has three attractors. It is $\Omega=\{\mathcal{C}_1,\mathcal{C}_2,\mathcal{C}_3\}$, where $\mathcal{C}_1=\{\delta_{4^6}^{1387}\}$, $\mathcal{C}_2=\{\delta_{4^6}^{2752}\}$ and $\mathcal{C}_3=\{\delta_{4^6}^{4096}\}$. The transient period of system \eqref{ex3-eq3} is $\tau=11$.
It is not hard to check that
\begin{equation}\label{ex3-eq5}
Col(L^{11})\nsubseteq\Lambda^\prime,
\end{equation}
which means that systems \eqref{ex3-eq1} and \eqref{ex3-eq2} are not globally completely synchronous.
Further, we can get that the maximum synchronization basin of systems \eqref{ex3-eq1} and \eqref{ex3-eq2} satisfies $|\Phi^\prime_{\max}|=2576<4^6$. That is to say, systems \eqref{ex3-eq1} and \eqref{ex3-eq2} are completely synchronous with respect to $\Phi^\prime_{\max}$.
Moreover, the shortest synchronization time with respect to $\Phi^\prime_{\max}$ is $\Gamma^\prime(\Phi^\prime_{\max})=10$.
Hence, the state trajectories of NFSR $1$ and NFSR $2$ starting from $\Phi^\prime_{\max}$ are exactly consistent after $10$ steps, which further indicates that the state trajectories of NFSR $1$ and NFSR $2$ starting from $\Delta_{4^6}\setminus\Phi^\prime_{\max}$ are different.

It should be pointed out that the investigation about complete synchronization of a cascade connection of two NFSRs is meaningful.  If two NFSRs in a cascade connection are completely synchronous, then the output sequences of two NFSRs are completely identical after a certain moment, which means that two NFSRs are equivalent. Reference \cite{zhongjianghua2019NFSR} has studied the equivalence of cascade connections of two NFSRs, which shows that finding properties of equivalent cascade connections of two NFSRs plays a great role in the design of the Grain family of stream ciphers.
\end{example}

\begin{example}\label{ex4}
In this example, we can consider the complete synchronization of an epigenetic model of gene regulation by protein-DNA interaction, which was introduced in \cite{1963Temporal, Heidel2003biochemicalsystems}. This model is composed of genetic locus $X_i$,  cellular structure $Y_i$, cellular locus $Z_i$, $i=1,2$ and its dynamic is given as follows:
\begin{equation}\label{ex4-1}
\begin{cases}
X_1(t)=\neg Z_1(t)\wedge \neg Z_2(t),\\
Y_1(t+1)=X_1(t),\\
Z_1(t+1)=Y_1(t),
\end{cases}
\end{equation}
and
\begin{equation}\label{ex4-2}
\begin{cases}
X_2(t)=\neg Z_2(t)\wedge \neg Z_1(t),\\
Y_2(t+1)=X_2(t),\\
Z_2(t+1)=Y_2(t).
\end{cases}
\end{equation}
Denote $\xi_1(t)=X_1(t)Y_1(t)Z_1(t)\in\Delta_{2^3}$, $\xi_2(t)=X_2(t)Y_2(t)Z_2(t)\in\Delta_{2^3}$ and $\xi(t)=\xi_1(t)\xi_2(t)\in\Delta_{2^6}$. We get the following system:
\begin{equation}\label{ex4-3}
  \xi(t+1)=L\xi(t),
\end{equation}
where $L=\delta_{64}$ $[37~37~38$ $38~\cdots~$ $63~27$ $64~28]$ $\in$ $\mathcal{L}_{64\times64}$.
According to \eqref{eq11}, one can get that the synchronous state set is
 $\Lambda^\prime=\{\delta_{64}^1, \delta_{64}^{10}, \delta_{64}^{19}, \delta_{64}^{28},  \delta_{64}^{37}, \delta_{64}^{46}, \delta_{64}^{55}, \delta_{64}^{64}\}$.

A simple calculation shows that the transient period of system \eqref{ex4-3} is $\tau=3$.
Combining with  Theorem \ref{th1}, one can get that
\begin{align}
Col(L^3)\subseteq \Lambda^\prime.
\end{align}
Therefore, systems \eqref{ex4-1} and \eqref{ex4-2} are completely synchronous, which is consistent with the conclusion obtained in reference \cite{chenhongwei2020LocalSynchronization}.

However, the shortest synchronization time is not specified in \cite{chenhongwei2020LocalSynchronization}.
Using Theorem \ref{pro4}, we derive that the shortest synchronization time of systems \eqref{ex4-1} and \eqref{ex4-2} is $\Gamma=3$.
\end{example}

\section{Conclusion}\label{sec7}

In this paper, the approximate synchronization problem  of two coupled  MVLNs has been addressed. Based on approximate synchronous state set, several necessary and sufficient criteria have been presented.
An effective algorithm has been developed for finding the maximum approximate synchronization basin.
And the method for calculating the shortest synchronization time has been provided.
On the other hand,  to make two systems achieve global approximate synchronization, pinning control strategy has been considered and three algorithms have been established, which  provide a constructive procedure for determining pinning nodes and designing pinning state feedback controllers.
In addition, relevant results have been used to analyze the complete synchronization of $k$-valued ($k\geq2$) logical networks.
Finally, the validity of the main results has been illustrated by four examples.

Note that the notion of approximate synchronization given in Definition 4 is based on the maximum error between corresponding state nodes. More generally, we can give another definition of approximate synchronization in average sense, which is replacing \eqref{eqrem1} by
\begin{equation*}
  \frac{1}{n}\left(\sum_{i=1}^n|x_i(t;X_0,Z_0)-z_i(t;X_0,Z_0)|\right)\leq \frac{\gamma}{k-1},~t\geq \rho.
\end{equation*}
Evidently, Definition 4 is a special case of the definition in average sense. Hence, it will be interesting to further investigate the approximate synchronization of average sense in the future work.
Besides, as Remark \ref{rem2} emphasizes, the pinning state feedback controllers designed in this paper are not unique. It is also significant to further investigate how to design pinning controllers for  minimizing the approximate synchronization time or the number of pinning nodes in the future work.

\section*{Appendix}
The proofs of Section IV are presented as follows.

\begin{proof}[{\bf  Proof of Theorem 1}]
1) First, we prove that  $\Omega_\Phi$, the set of all attractors of $\Phi$, can be expressed as
\begin{equation}
\Omega_\Phi=\bigcup_{t=\tau}^{\tau+\lambda-1}\{Col_i(L^t): \delta_{k^{2n}}^i\in\Phi\}.
\end{equation}

In fact, according to Lemma \ref{lem3}, for any state $\delta_{k^{2n}}^i\in\Phi$, the attractor of $\delta_{k^{2n}}^i$ is
  \begin{align}
    \Omega_{\delta_{k^{2n}}^i}&=\{L^t\delta_{k^{2n}}^i: t\geq\tau\}\notag\\
                             &=\{Col_i(L^t): t\geq\tau\}\notag\\
                             &=\{Col_i(L^t): \tau\leq t\leq\tau+\lambda-1\}.
  \end{align}
Hence,  $\Omega_\Phi$ can be expressed as
 \begin{align}
\Omega_\Phi&=\bigcup_{\delta_{k^{2n}}^i\in\Phi}\{Col_i(L^t): \tau\leq t\leq\tau+\lambda-1\}\notag\\
             &=\bigcup_{t=\tau}^{\tau+\lambda-1}\{Col_i(L^t): \delta_{k^{2n}}^i\in\Phi\}.
 \end{align}

(Sufficiency) Assume  $\bigcup_{t=\tau}^{\tau+\lambda-1}\{Col_i(L^t): \delta_{k^{2n}}^i\in\Phi\}\subseteq\Lambda$.
Then $\Omega_\Phi\subseteq\Lambda$.  It means that for any initial state $\xi_0\in\Phi$, $\xi(t;\xi_0)\in\Omega_\Phi\subseteq\Lambda$ holds for any $t\geq\tau$. Hence, systems \eqref{sys3} and \eqref{sys4} are  approximately synchronous with respect to $\Phi$.

(Necessity) Suppose that  systems \eqref{sys3} and \eqref{sys4} are approximately synchronous with respect to $\Phi$.
According to Corollary \ref{th0}, there exists an integer $\rho\in\mathbb{Z}_+$ such that $\xi(t;\xi_0)\in\Lambda$ holds for any $\xi_0\in\Phi$ and $t\geq\rho$.
If equation \eqref{eq4} is invalid, then there exists an attractor $\mathcal{C}_l\in\Omega_\Phi$ such that $\mathcal{C}_l\nsubseteq\Lambda$. However, there exists $\xi^\prime_0\in\Phi$ such that $\xi(t;\xi^\prime_0)\in\mathcal{C}_l$ holds for any $t\geq\tau$, which is a contradiction. Therefore, the necessity is established.

2) It is noticed that if $\Phi=\Delta_{k^{2n}}$, then the set of attractors of $\Phi$
can be expressed as $\Omega_\Phi=\Omega=Col(L^\tau)$. With this in mind,  conclusion 2) is obvious based on the proof of conclusion 1).
\end{proof}

\begin{proof}[{\bf Proof of Corollary 2}]
In light of the definition of $\Xi_1$, we see that $L^t\Xi_1=\sum_{\delta_{k^{2n}}^i\in\Phi}Col_i(L^t)$. Consequently, we have
\begin{align*}
&sgn\left(\left(\sum_{t=\tau}^{\tau+\lambda-1}L^t\right)\Xi_1\right)\\
&=sgn\left(\sum_{t=\tau}^{\tau+\lambda-1}L^t\Xi_1\right)\\
&=sgn\left(\sum_{t=\tau}^{\tau+\lambda-1}\sum_{\delta_{k^{2n}}^i\in\Phi}Col_i(L^t)\right).
\end{align*}
Then $\left[sgn\left(\left(\sum_{t=\tau}^{\tau+\lambda-1}L^t\right)\Xi_1\right)\right]_{l,1}=1$ is equivalent to $\delta_{k^{2n}}^l$ $\in$ $\bigcup_{t=\tau}^{\tau+\lambda-1}\{Col_i(L^t): \delta_{k^{2n}}^i\in\Phi\}$.
From the definition of $\Xi_2$, it is also clear that $[\Xi_2]_{l,1}=1$ is equivalent to $\delta_{k^{2n}}^l\in\Lambda$.
Hence, $\bigcup_{t=\tau}^{\tau+\lambda-1}\{Col_i(L^t): \delta_{k^{2n}}^i\in\Phi\}\subseteq\Lambda$ is equivalent to $sgn\left(\left(\sum_{t=\tau}^{\tau+\lambda-1}L^t\right)\Xi_1\right)\leq\Xi_2$.
Similarly,  $Col(L^\tau)\subseteq\Lambda$ is equivalent to $sgn(L^\tau\mathbf{1}_{k^{2n}})\leq\Xi_2$.
According to Theorem \ref{th1}, the proof is completed.
\end{proof}

\begin{proof}[{\bf Proof of Algorithm 1}]
According to the Theorem 2, systems \eqref{sys3} and \eqref{sys4} are approximately synchronous with respect to $\Phi$ if and only \eqref{eqth2} holds. Note that if the set $\Phi$ is taken as a singleton set, i.e., $\Phi=\{\delta_{k^{2n}}^i\}$, then \eqref{eqth2} becomes
\begin{equation}\label{eqal1}
sgn\left(\sum_{t=\tau}^{\tau+\lambda-1}Col_i\left(L^t\right)\right)\leq\Xi_2.
\end{equation}
Therefore, systems \eqref{sys3} and \eqref{sys4} are approximately synchronous with respect to $\Phi=\{\delta_{k^{2n}}^i\}$ if and only \eqref{eqal1} holds.
Hence, the maximum approximate synchronization basin contains all states satisfying condition \eqref{eqal1}, i.e., Algorithm 1 is valid.
\end{proof}

\begin{proof}[{\bf Proof of Algorithm 2}]
From the definition of invariant subset, we can see that $S$ is an invariant subset of system \eqref{sys5} if and only if $L_S\subseteq S$, where $L_S:=\{Col_i(L): \delta_{k^{{2n}}}^i\in S\}$.
Consequently, $\delta_{k^{2n}}^i$ belongs to the maximum invariant subset set of $\Lambda$ if and only if $Col_i(L)\in\Lambda$. Hence, correctness of Algorithm 2 is clear.
\end{proof}

\begin{proof}[{\bf Proof of Algorithm 3}]
Note that our purpose is to obtain the attractors that are not in $\Lambda$ for the  design of pinning controllers. Combined with Hadamard product of matrices, it follows from \eqref{eq7-1} and \eqref{eq7-2} that $\Xi^{i}\circ\Xi_2=\Xi^{i}$ is equivalent to $\mathcal{C}_i\in\Lambda$. Therefore, if $\mathcal{C}_i\in\Omega_0$, then $\Xi^{i}\circ\Xi_2\neq\Xi^{i}$. Apparently, the efficiency of Algorithm 3 is clear.
\end{proof}

\begin{proof}[{\bf Proof of Theorem \ref{pro4}}]
1) Assume systems \eqref{sys3} and \eqref{sys4} are approximately synchronous with respect to $\Phi$. Then Theorem \ref{th1} shows that all attractors of $\Phi$ belong to $\Lambda$, i.e., $\Omega_\Phi\subseteq\Lambda$.
Denoting $\varpi:=\mathop{\arg\min}_{1\leq t\leq \tau}\{sgn(L^t\Xi_1)\leq\Xi_3\}$, we prove that $\Gamma_{\Phi}=\varpi$.

On the one hand, since $\tau$ is the transient period of system \eqref{sys5},
the trajectory of system \eqref{sys5} starting from any initial state $\xi_0\in\Phi$ will eventually reach into an attractor of $\Phi$ after the  moment $\tau$,
i.e.,  $\xi(t;\xi_0)\in\Omega_\Phi\subseteq\Lambda$ holds for any $t\geq\tau$ and $\xi_0\in\Phi$. It means that systems \eqref{sys3} and \eqref{sys4} have achieve approximate synchronization when $t=\tau$. Note that $\Gamma_{\Phi}$ is the shortest approximate synchronization time. Hence, $1\leq \Gamma_{\Phi}\leq \tau$.

On the other hand, condition
\begin{equation}
sgn(L^t\Xi_1)\leq\Xi_3
\end{equation}
implies that for any $\xi_0\in\Phi$, $\xi(t;\xi_0)\in \mathcal{I}_m^\Lambda$. Therefore, $t=\varpi$ is the smallest integer such that $\xi(\varpi;\xi_0)\in \mathcal{I}_m^\Lambda$ holds for any $\xi_0\in\Phi$.
Combining the definition of maximum invariant subset, one sees that  $\varpi$ is the smallest integer such that  $\xi(t;\xi_0)\in\Lambda$ holds for any $t\geq\varpi$ and $\xi_0\in\Phi$. That is $\Gamma_{\Phi}=\varpi$.


2) Assume systems \eqref{sys3} and \eqref{sys4} are globally approximately synchronous.
It is noted that if $\Phi=\Delta_{k^{2n}}$, then $\Xi_1=\mathbf{1}_{k^{2n}}$.
As a result, we can conclude that the globally  shortest approximate synchronization time $\Gamma$ satisfies \eqref{eqpro4-2} by the proof of conclusion 1).
\end{proof}

\begin{proof}[{\bf Proof of Theorem 3}]
For the case of $\Phi_{\max}=\emptyset$, after the process of lines 3-9 in Algorithm 4, it is clear that $L_\Lambda:=\{Col_i(\overline{L}): \delta_{k^{{2n}}}^i\in \Lambda\}\subseteq \Lambda$, which means that
the approximately synchronous state set $\Lambda$ becomes an invariant subset of $\Delta_{k^{2n}}$. Lines 10-13 in Algorithm 4 guarantee that all states in $\Delta_{k^{2n}}$ can reach into $\Lambda$ after finite steps.
In other words,  there exists an integer $\rho_1\in\mathbb{Z}_+$  such that for any initial state $\xi_0\in\Delta_{k^{2n}}$, $\xi(\rho_1;\xi_0)\in\Lambda$.
Consequently, we derive that for any $\xi_0\in\Delta_{k^{2n}}$,
\begin{equation}\label{eqth2-1}
\xi(t;\xi_0)\in\Lambda, ~t\geq\rho_1.
\end{equation}
For the case of $\Phi_{\max}\neq\emptyset$, one sees that $\mathcal{I}_m^\Lambda\neq\emptyset$.
According to lines 15-18 in Algorithm 4, all states in  $\Delta_{k^{2n}}$ can reach into $\mathcal{I}_m^\Lambda$ after finite steps. That is to say,
there exists an integer $\rho_2\in\mathbb{Z}_+$ such that for any $\xi_0\in\Delta_{k^{2n}}$,
\begin{equation}\label{eqth2-2}
\xi(t;\xi_0)\in \mathcal{I}_m^\Lambda\subseteq\Lambda,~ t\geq\rho_2.
\end{equation}
On the basis of Corollary \ref{th0}, it follows from \eqref{eqth2-1} and \eqref{eqth2-2} that systems \eqref{sys3} and \eqref{sys4} can achieve global approximate synchronous. The proof is completed.
\end{proof}

\begin{proof}[{\bf Proof of Proposition \ref{pro2}}]
It is noticed that equation \eqref{eq6} makes sense if and only if $m^i_{(w^i_l-1)k+h^i_l}$, $l\in\{1,2,\ldots,k^{2n}\}$ can be uniquely determined by $w^i_l$. Obviously, only in the case of $h^i_{l_1}=h^i_{l_2}, \overline{h}^i_{l_1}\neq\overline{h}^i_{l_2}$, the value of $m^i_{(w^i_l-1)k+h^i_l}$ can not be determined uniquely, which is needed to discuss.

Based on \eqref{eq9} and \eqref{eq10}, it is easy to see that $0\leq|\mathfrak{T}_\sigma^i|\leq k$.

1) If $|\mathfrak{T}_\sigma^i|\leq1$, $\sigma=1,2,\ldots,k$, one can see that once the value of $w^i_l$ is fixed, then the value of $m^i_{(w^i_l-1)k+h^i_l}$ is equal to $\overline{h}^i_l$. That is $w^i_l$ can take any value from $\{1, 2,\dots,k\}$.

2) Assume there exists $\sigma^\prime\in\{1,2,\ldots,k\}$ such that $2\leq|\mathfrak{T}_{\sigma\prime}^i|\leq k$. We denote $\mathfrak{T}_{\sigma^\prime}^i:=\{T_{\sigma,\beta_1}^i,T_{\sigma,\beta_2}^i,\ldots,
T_{\sigma,\beta_{|\mathfrak{T}_\sigma^i|}}^i\}$.
On the contrary, if there exist $l_{\mu}\in T_{\sigma^\prime,\beta_\mu}^i$, $l_{\nu}\in T_{\sigma^\prime,\beta_\nu}^i$  and $\mu\neq\nu$ such that
$w^i_{l_\mu}=w^i_{l_\nu}$.
Then one can get from \eqref{eq6} that
\begin{equation}
\beta_\mu=m^i_{(w^i_{l_\mu}-1)k+\sigma^\prime}=m^i_{(w^i_{l_\nu}-1)k+\sigma^\prime}=\beta_\nu,
\end{equation}
but $\beta_\mu\neq\beta_\nu$. It is a contradiction.
Hence, the conclusion is established.
\end{proof}

\section*{References}

%

\end{document}